\documentclass[10pt,journal]{IEEEtran}

\makeatletter
\def\ps@headings{%
\def\@oddhead{\mbox{}\scriptsize\rightmark \hfil \thepage}%
\def\@evenhead{\scriptsize\thepage \hfil \leftmark\mbox{}}%
\def\@oddfoot{}%
\def\@evenfoot{}}
\makeatother 

\usepackage {amsmath}
\usepackage {amssymb}
\usepackage {latexsym}
\usepackage {graphicx}
\usepackage {cite}
\usepackage {subfigure}
\usepackage {url}
\usepackage {stfloats}
\usepackage {mathrsfs}
\usepackage{setspace}

\newtheorem{theorem}{Theorem}

\begin{document}

\IEEEoverridecommandlockouts

\newcommand{\ls}[1]
    {\dimen0=\fontdimen6\the\font
     \lineskip=#1\dimen0
     \advance\lineskip.5\fontdimen5\the\font
     \advance\lineskip-\dimen0
     \lineskiplimit=.9\lineskip
     \baselineskip=\lineskip
     \advance\baselineskip\dimen0
     \normallineskip\lineskip
     \normallineskiplimit\lineskiplimit
     \normalbaselineskip\baselineskip
     \ignorespaces
}

\title{Base-Station Selections for QoS Provisioning \\Over
Distributed Multi-User MIMO Links \\in Wireless Networks}
\author{\authorblockN{\vspace{-0pt}Qinghe Du and Xi Zhang}\vspace{-0pt}\\
\authorblockA{Networking and Information Systems Laboratory\\
Department of Electrical and Computer Engineering\\
Texas A\&M University, College Station, TX 77843, USA\\
Email: {\it \{duqinghe@tamu.edu,
xizhang@ece.tamu.edu\}}\vspace*{-0pt}}
\thanks{The research reported in this paper was supported in part by the U.S.
National Science Foundation CAREER Award under Grant ECS-0348694.}}
 \IEEEpeerreviewmaketitle \maketitle

\markboth{TO APPEAR in Proceeding of IEEE INFOCOM 2011}{TO APPEAR in
Proceeding of IEEE INFOCOM 2011}

\begin{abstract}
We propose the QoS-aware BS-selection and the corresponding
resource-allocation schemes for downlink multi-user transmissions
over the distributed multiple-input-multiple-output (MIMO) links,
where multiple location-independent base-stations (BS), controlled
by a central server, cooperatively transmit data to multiple mobile
users. Our proposed schemes aim at minimizing the BS usages and
reducing the interfering range of the distributed MIMO
transmissions, while satisfying diverse statistical delay-QoS
requirements for all users, which are characterized by the
delay-bound violation probability and the effective capacity
technique. Specifically, we propose two BS-usage minimization
frameworks to develop the QoS-aware BS-selection schemes and the
corresponding wireless resource-allocation algorithms across
multiple mobile users. The first framework applies the joint
block-diagonalization (BD) and probabilistic transmission (PT) to
implement multiple access over multiple mobile users, while the
second one employs time-division multiple access (TDMA) approach to
control multiple users' links. We then derive the optimal
BS-selection schemes for these two frameworks, respectively. In
addition, we further discuss the PT-only based BS-selection scheme.
Also conducted is a set of simulation evaluations to comparatively
study the average BS-usage and interfering range of our proposed
schemes and to analyze the impact of QoS constraints on the BS
selections for distributed MIMO transmissions.
\end{abstract}

\vspace{5pt}\begin{keywords} Distributed MIMO, broadband wireless
networks, statistical QoS provisioning, wireless fading
channels.\vspace{-5pt}
\end{keywords}


\section{Introduction}
\label{sect-introduction}

\PARstart{T}{o} increase the coverage of broadband wireless
networks, distributed multiple-input-multiple-output (MIMO)
techniques, where multiple location-independent base stations (BS)
cooperatively transmit data to mobile users, have attracted more and
more research attentions~\cite{A-Sanderovich,P-Shang,R-Mudumbai}. In
particular, the distributed MIMO techniques can effectively organize
multiple location-independent BS's to form the distributed MIMO
links connecting with mobile users
Like the conventional centralized MIMO
system~\cite{E-Telatar,M-Gharavi-Alkhansari,S-Sanayei}, the
distributed MIMO system can significantly enhance the capability of
the broadband wireless networks in terms of the quality-of-service
(QoS) provisioning as compared to the single antenna system.
However, the distributed nature for cooperative multi-BS
transmissions also imposes many new challenges in wide-band wireless
communications, which are not encountered in the centralized MIMO
systems. First, the cooperative distributed transmissions cause the
severe difficulty for synchronization among multiple
location-independent BS transmitters. Second, as the number of
cooperative BS's increases, the computational complexity for MIMO
signal processing and coding also grow rapidly. Third, because the
coordinated BS's are located at different geographical positions,
the cooperative communications in fact enlarge the interfering areas
for the used spectrum, thus drastically degrading the
frequency-reuse efficiency in the spatial domain. Finally, many
wide-band transmissions are sensitive to the delay, and thus we need
to design QoS-aware distributed MIMO techniques, such that the
scarce wireless resources can be more efficiently utilized.


Towards the above issues, many research works on distributed MIMO
transmissions have been proposed recently. The feasibility of
transmit beamforming with efficient synchronization techniques over
distributed MIMO link has been demonstrated through experimental
tests~\cite{R-Mudumbai}, suggesting that complicated MIMO signal
processing techniques are promising to implement in realistic
systems. For the centralized MIMO system, the antenna
selection~\cite{S-Sanayei,M-Gharavi-Alkhansari} is an effective
technique to reduce the complexity, which clearly can be also
extended to distributed MIMO systems for the BS selection. It can be
expected that the BS-selection techniques can significantly decrease
the processing complexity, while still achieving high throughput
gain over the single BS transmission. Also, it is desirable to
minimize the number of selected BS's through BS-selection
techniques, which can effectively decrease the interfering range and
thus improve the frequency-reuse efficiency of the entire wireless
network. Most previous research works for BS/antennas selections
mainly focused on the scenarios of selecting a subset of
BS's/antennas with the fixed
cardinality~\cite{S-Sanayei,M-Gharavi-Alkhansari,P-Shang}. However,
it is evident that based on the wireless-channel status, BS-subset
selections with dynamically adjusted cardinality can further
decrease the BS usage. More importantly, how to efficiently support
diverse delay-QoS requirements through BS-selection in distributed
MIMO systems sill remains a widely cited open problem.

To overcome the aforementioned problems, we propose the QoS-aware
BS-selection schemes for the distributed wireless MIMO links, which
aim at minimizing the BS usages and reducing the interfering range,
while satisfying diverse statistical delay-QoS constraints. In
particular, we develop two BS-usage minimization frameworks for
distributed multi-suer MIMO transmissions. The first framework uses
the joint block-diagonalization (BD) and probabilistic transmission
(PT) for multiple access of multi-user over distributed MIMO links,
while the second framework employs time-division multiple access
(TDMA) techniques. We derive the optimal QoS-aware BS-selection and
the corresponding resource allocation schemes for these two
frameworks, respectively. We also discuss the PT-only based
BS-selection scheme. Simulations are conducted for comparative
analyses among the above BS-selection schemes.

%
%


The rest of this paper is organized as follows.
Section~\ref{sect-sysmodel} describes the system model for
distributed MIMO transmissions.
Section~\ref{sect-effective-capacity} introduces the statistical QoS
guarantees and the concept of effective capacity.
Section~\ref{sect-multiple-BD} develops the joint BD and PT (BD-PT)
optimization frameworks for QoS-aware BS-selections over multi-user
distributed MIMO links and derives the corresponding optimal
solution. Section~\ref{sect-multiple-TDMA} derives TDMA-based
QoS-aware BS-selection scheme. Section~\ref{sect-simulations}
conducts simulations to perform comparative analyses for our
proposed schemes. The paper concludes with
Section~\ref{sect-conclusion}.

\emph{Notations:} The operator $|\cdot|$ used on a real or complex
number generates the absolute value; the operator $|\cdot|$ used for
a set represents the cardinality of this set. We use boldface to
denote matrices and vectors. For an $X\times Y$ matrix $\mathbf{A}$,
we denote by $\mathbf{A}(i,j)$ the element on the $i$th row and
$j$th column; $\|\mathbf{A}\|_F$ denotes the Frobenius norm of
$\mathbf{A}$, where $\|\mathbf{A}\|_F\triangleq
\sqrt{\sum_{i=1}^X\sum_{j=1}^Y|\mathbf{A}(i,j)|^2}$. The operators
$(\cdot)^{\tau}$ and $(\cdot)^{\dag}$ generate the transpose and
conjugate transpose, respectively. The operator $1_{(\cdot)}$ is the
indication function. If the statement in the subscript is true, we
have $1_{(\cdot)}=1$; otherwise, $1_{(\cdot)}=0$.

\section{System Model}
\label{sect-sysmodel}

\subsection{System Architecture}
\label{sect-sysmodel-architecture}

We concentrate on the wireless \emph{distributed MIMO} system for
downlink transmissions depicted in Fig.~\ref{fig-sysmodel}, which
consists of $K_{\mathrm{bs}}$ distributed BS's, $K_{\mathrm{mu}}$
mobile users, and one central server. The $m$th BS has $M_m$
transmit antennas for $m=1,2,\ldots,K_{\mathrm{bs}}$ and the $n$th
mobile user has $N_n$ receive antennas for
$n=1,2,\ldots,K_{\mathrm{mu}}$. All distributed BS's are connected
to the central server through high-speed optical connections. The
data to be delivered to the $n$th mobile user,
$n=1,2,\ldots,K_{\mathrm{mu}}$, arrives at the central server with a
constant rate, which is denoted by $\overline{C}_n$. Then, the
central server dynamically controls these distributed BS's to
cooperatively transmit data to the corresponding mobile users under
the specified delay-QoS requirements.

\begin{figure}[t]
\vspace{1pt}\centerline{\includegraphics[width=2.8in]{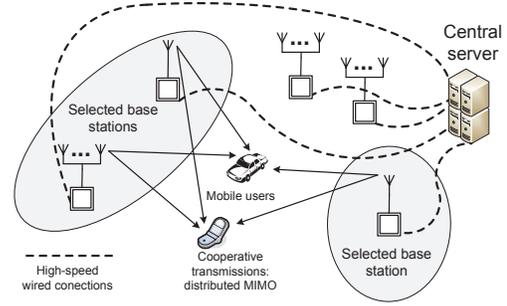}}
\vspace{-5pt}\caption{System model for wireless downlink distributed
MIMO transmissions.} \vspace{-15pt}\label{fig-sysmodel}
\end{figure}

For multi-user downlink transmissions, the distributed BS's and the
mobile users form the broadcast MIMO link for data transmissions.
The wireless fading channels between the $m$th BS and the $n$th
mobile user is modeled by an $N_n\times M_m$ matrix
$\mathbf{H}_{n,m}$, where $\mathbf{H}_{n,m}(i,j)$ is the complex
channel gain between the $i$th receive antenna of $n$th mobile user
and the $j$th transmit antenna of the $m$th BS. All elements of
$\mathbf{H}_{n,m}$ are independent and circularly symmetric complex
Gaussian random variables with zero mean and the variance equal to
$\overline{h}_{n,m}$, implying that $\mathbf{H}$ has continuous
cumulative distribution function (CDF). Also, the instantaneous
aggregate power gain of the MIMO link between the $n$th mobile user
and the $m$th BS, denoted by $\gamma_{n,m}$, is defined by
\begin{eqnarray}
\gamma_{n,m} \triangleq \frac{1}{M_m}
\left\|\mathbf{H}_{n,m}\right\|_F^2
\label{eq-aggregate-power-gain}
\end{eqnarray}
Since the Frobenius norm of the channel matrix can effectively
characterize the channel quality in terms of achieving high
throughput~\cite{S-Sanayei}, the aggregate power gain given in
Eq.~(\ref{eq-aggregate-power-gain}) will play an important role in
our BS selection design. We further define $\mathbf{H}_n\triangleq
[\mathbf{H}_{n,1}~\mathbf{H}_{n,2}~\cdots~\mathbf{H}_{n,K_{\mathrm{bs}}}]$
as the CSI for the $n$th mobile user for
$n=1,2,\ldots,K_{\mathrm{mu}}$. The matrix $\mathbf{H}_n$ follows
the independent block-fading model, where $\mathbf{H}_n$ does not
change within a time period with the fixed length $T$, called a time
frame, but varies independently from one frame to the other frame.
Furthermore, we define $\mathbf{H}\triangleq
[\mathbf{H}_1^{\tau}~\mathbf{H}_2^{\tau}~\cdots~\mathbf{H}_{K_{\mathrm{bs}}}^{\tau}]^{\tau}$,
representing a fading state of the entire distributed MIMO system.
%

In order to decrease the complexity and suppress the interfering
range of the distributed MIMO transmission, the cental server
dynamically selects a subset of BS's to construct the distributed
MIMO link. Then, our design target is to minimize the average number
of needed BS's subject to the specified QoS constraints. We suppose
that each mobile user can perfectly estimate its CSI at the
beginning of every time frame and reliably feed CSI back to the
central server through dedicated control channels. Based on CSI
$\mathbf{H}$ and QoS requirements, the central server then
adaptively selects the subset of BS's and organizes them to transmit
data to mobile users through the distributed MIMO links.

\subsection{The Delay QoS Requirements}

The central data server maintains a queue for the incoming traffic
to each mobile user. We mainly focus on the queueing delay in this
paper because the wireless channel is the major bottleneck for
high-rate wireless transmissions. Since it is usually unrealistic to
guarantee the hard delay bound over the highly time-varying wireless
channels, we employ the statistical metric, namely, the
\emph{delay-bound violation probability}, to characterize the
diverse delay QoS requirements. Specifically, for the $n$th mobile
user, the probability of violating a specified delay bound, denoted
by $D_{\mathrm{th}}^{(n)}$, cannot exceed a given threshold $\xi_n$.
That is, the inequality
\begin{eqnarray}
\mbox{Pr}\left\{D_n>D_{\mathrm{th}}^{(n)}\right\}\leq \xi_n,\quad
n=1,2,\ldots,N_{\mathrm{mu}}, \label{eq-statistical-delay-qos}
\end{eqnarray}
needs to hold, where $D_n$ denotes the queueing delay in the $n$th
mobile user's queueing system.

\subsection{Performance Metrics and Design Objective}
\label{sect-sysmodel-design-objective}

We denote by $L$ the cardinality of the selected BS subset (the
number of selected BS's) for the distributed MIMO transmission in a
fading state. Then, we denote the expectation of $L$ by
$\overline{L}$ and call it the \emph{average BS usage}. As mentioned
in Section~\ref{sect-sysmodel-architecture}, our major objective is
to minimize $\overline{L}$ through dynamic BS selection while
guaranteeing the delay QoS constraint specified by
Eq.~(\ref{eq-statistical-delay-qos}). We will also evaluate the
\emph{average interfering range} affected by the distributed MIMO
transmission. The instantaneous interfering range, denoted by $A$,
is defined as the area of the region where the average received
power under the current MIMO transmission is larger than then
certain threshold denoted by $\sigma^2_{\mathrm{th}}$. The average
interfering area is then defined as the expectation
$\mathbb{E}_{\mathbf{H}}\{A\}$ over all $\mathbf{H}$. Clearly,
minimizing $\overline{L}$ can not only reduce implementation
complexity, but also decrease the average interfering range affected
by the transmit power.

%

\subsection{The Power Control Strategy}
\label{sect-sysmodel-power}

The transmit power of our distributed MIMO system varies with the
number of selected BS's. In particular, given the number $L$ of
selected BS's, the total instantaneous transmitted power used for
distributed MIMO transmissions is set as a constant equal to
$\mathcal{P}_L$. Furthermore, $\mathcal{P}_L$ linearly increases
with $L$ by using the strategy as follows:
\begin{eqnarray}
\mathcal{P}_L = \mathcal{P}_{\mathrm{ref}} + \kappa (L-1) ,\quad
L=1,2,\ldots, K_{\mathrm{bs}}, \label{eq-power-linear-increase}
\end{eqnarray}
where $\mathcal{P}_{\mathrm{ref}}>0$ is called the \emph{reference
power} and $\kappa\geq0$ describes the power increasing rate against
$L$. Also, we define $\mathcal{P}_L\triangleq0$ for $L=0$. The above
power adaptation strategy is simple to implement, while the average
transmit power can be effectively decreased through minimizing the
average number of used BS's. In addition,
Eq.~(\ref{eq-power-linear-increase}) can upper-bound the
instantaneous interferences and the interfering range over the
entire network.
%
%

\section{Effective Capacity Approach for Statistical
Delay-QoS Guarantees} \label{sect-effective-capacity}

In this paper, we apply the effective capacity
approach~\cite{D-Wu,J-Tang,X-Zhang,J-Tang-twc-jun-2008} to integrate
the constraint on delay-bound violation probability given by
Eq.~(\ref{eq-statistical-delay-qos}) into our BS selection design.
Consider a stable discrete-time queueing system with the stationary
time-varying arrival-rate and departure-rate (service-rate)
processes. The asymptotic analyses based on the large deviation
principal~\cite{D-Wu,C-S-Chang-book} show that under the sufficient
conditions, the probability that the queue-length, denoted by $Q$,
exceeding a given bound $Q_{\mathrm{th}}$ can be approximated by
\begin{eqnarray}
\mbox{Pr}\{Q>Q_{\mathrm{th}}\}  \approx e^{-\theta Q_{\mathrm{th}}},
\label{eq-prob-queue-vio}
\end{eqnarray}
where $\theta>0$ is a constant called \emph{QoS exponent}. It is
clear that the larger (smaller) $\theta$ implies the lower (higher)
queue-length-bound violation probability.

By using $\theta$, the delay-bound violation probability can be
approximated~\cite{D-Wu} by
\begin{eqnarray}
\mbox{Pr}\{D>D_{\mathrm{th}}\}  \approx e^{-\theta \overline{C}
D_{\mathrm{th}}}. \label{eq-prob-delay-vio}
\end{eqnarray}
for the constant rate $\overline{C}$. When the arrival rate is not
time-varying, the approximation in Eq.~(\ref{eq-prob-delay-vio})
needs to replace $\overline{C}$ with effective
bandwidth~\cite{D-Wu,C-S-Chang-book} function of the arrival rate
process, which is defined as the minimum constant service rate
required to guarantee QoS exponent $\theta$.


Then, to upper-bound $\mbox{Pr}\{D>D_{\mathrm{th}}\}$ with a
threshold $\xi$, using Eq.~(\ref{eq-prob-delay-vio}), we get the
minimum required QoS exponent $\theta$ as follows:
\begin{eqnarray}
\theta = -\frac{\log(\xi)}{\overline{C}D_{\mathrm{th}}}.
\label{eq-desired-theta}
\end{eqnarray}
Consider a discrete-time arrival process with constant rate
$\overline{C}$ and a discrete-time time-varying stationary departure
process, denoted by $R[k]$, where $k$ is the time index. In order to
guarantee the desired $\theta$ determined by
Eq.~(\ref{eq-desired-theta}), the statistical QoS
theory~\cite{C-S-Chang-book,D-Wu} shows that the \emph{effective
capacity} $\mathcal{C}(\theta)$ of the service-rate process $R[k]$
needs to satisfy
\begin{eqnarray}
\mathcal{C}(\theta)=\overline{C},\label{eq-desired-effcp}
\end{eqnarray}
given the QoS exponent $\theta$. The \emph{effective capacity}
function is defined in~\cite{D-Wu} as the maximum constant arrival
rate which can be supported by the service rate to guarantee the
specified QoS exponent $\theta$. If the service-rate sequence $R[k]$
is time uncorrelated, the effective capacity can be
written~\cite{J-Tang} as
\begin{eqnarray}
\mathcal{C}(\theta) \triangleq
-\frac{1}{\theta}\log\left(\mathbb{E}\left\{e^{-\theta
R[k]}\right\}\right),\label{eq-effective-capacity}
\end{eqnarray}
where $\mathbb{E}\{\cdot\}$ denotes the expectation.

In our distributed MIMO system, the BS selection result is designed
as the function determined by the current CSI. Thus, the
corresponding transmission rate (service rate) is time independent
under the independent block-fading model (see
Section~\ref{sect-sysmodel-architecture}). Then, applying
Eqs.~(\ref{eq-desired-theta})-(\ref{eq-desired-effcp}), the delay
QoS constraints given by Eq.~(\ref{eq-statistical-delay-qos}) can be
equivalently converted to:
\begin{eqnarray}
\mathbb{E}_{\mathbf{H}}\left\{e^{-\theta_n R_n}
-e^{-\theta_n \overline{C}_n}\right\}\leq 0, \quad,
n=1,2,\ldots,N_{\mathrm{mu}}, \label{eq-statistical-delay-qos-1}
\end{eqnarray}
where
$\theta_n=-\log(\xi_n)\!\!\left/\!\big(\overline{C}_nD_{\mathrm{th}}^{(n)}\big)\right.$
and $\mathbb{E}_{\mathbf{H}}\{\cdot\}$ denotes the expectation over
all $\mathbf{H}$.


\section{Joint Block-Diagonalization and Probabilistic Transmission Based Base-Station Selection}
\label{sect-multiple-BD}

As discussed in Section~\ref{sect-sysmodel}, based on CSI
$\mathbf{H}$ and QoS requirements, the central server will
adaptively select the subset of BS's and organizes them to transmit
data to mobile users through the distributed MIMO links. Given the
cardinality $L$ of the desired BS subset in a fading state, we
denote by $\Omega_L$ the set of indices of selected BS's, where
$\Omega_L=\{i_{L,1},i_{L,2},\ldots,i_{L,L}\}$ and
$i_{L,\ell}\in\{1,2,\ldots,K_{\mathrm{bs}}\}$ for
$\ell=1,2,\ldots,L$. Note that once a BS is selected, we use all its
transmit antennas for data transmissions. For the specified $L$, we
use $\mathcal{U}_L=\{n_{U,1},n_{U,2},\ldots,n_{U,U}\}$ to denote the
set of \emph{active} users, picked by the central server, which can
receive the data in this fading state, where $U$ is the cardinality
of $\mathcal{U}_L$. For presentation convenience, we use
$\mathfrak{M}_L\triangleq(\Omega_L,\mathcal{U}_L)$ to represent a
specific \emph{transmission mode} (or mode in short). Moreover, we
term the pairs with $U\geq2$ for $\mathcal{U}_L$ as
\emph{multi-user} modes, and term the pairs with $U=1$ as
\emph{single-user} modes.

Given $L$, $\Omega_L$ and $\mathcal{U}_L$, the channel matrix of the
$n$th mobile user for $n\in\mathcal{U}_L$, modeled by
$\mathbf{H}_{\Omega_L}^{(n)}$, is determined by
\begin{eqnarray}
\mathbf{H}_{\Omega_L}^{(n)} \triangleq \left[\mathbf{H}_{n,i_{L,1}}~
\mathbf{H}_{n,i_{L,2}}~\cdots~ \mathbf{H}_{n,i_{L,L}} \right],
\end{eqnarray}
where $\mathbf{H}_{\Omega_L}^{(n)}$ is an
$N_n\times\big(\sum_{\ell\in\Omega_L} M_{i_{L,\ell}}\big)$ matrix.
Furthermore, we use
$\overline{\boldsymbol{\Upsilon}}_{\Omega_L}^{(n)}$ to denote the
power gain matrix for $\mathbf{H}_{\Omega_L}^{(n)}$ under the given
$\Omega_L$, where
\begin{eqnarray}
\overline{\boldsymbol{\Upsilon}}_{\Omega_L}^{(n)}(i,j)=\mathbb{E}_{\mathbf{H}}\left\{\left.\left|\mathbf{H}_{\Omega_L}^{(n)}(i,j)\right|^2\right|
\mbox{fixing}
~\Omega_L\right\}.\label{eq-multiple-BD-csi-avg-power-gain}
\end{eqnarray}
The physical-layer signal transmissions can be modeled by
\begin{eqnarray}
\begin{array}{l}
\mathbf{y}_{\mathfrak{M}_L}^{(n)} = \mathbf{H}_{\Omega_L}^{(n)}
\sum_{i\in\mathcal{U}_L} \mathbf{s}_{\mathfrak{M}_L}^{(i)} +
\boldsymbol{\varsigma}^{(n)}, \quad n\in\mathcal{U}_L,
\end{array}
\nonumber
\end{eqnarray}
where $\mathbf{s}_{\mathfrak{M}_L}^{(i)}$ represents
the $i$th user's input signal vector for the MIMO channel
$\mathbf{H}_{\Omega_L}^{(i)}$, $\mathbf{y}_{\mathfrak{M}_L}^{(n)}$
is the signal vector received by the $n$th user, and
$\boldsymbol{\varsigma}^{(n)}$ is the complex additive white
Gaussian noise (AWGN) vector with unit power for each element of
this vector. In this section, we employ the
\emph{block-diagonalization} (BD) technique~\cite{Q-H-Spencer} to
implement multiple access for multi-user modes in our QoS-aware
BS-selection framework.

For dynamic BS selections in distributed MIMO transmissions, $L$ and
$\mathfrak{M}_L$ are both functions of CSI and QoS requirements.
Then, we need to answer the following questions: (i)~Given $L$,
which transmission mode will be used for single-user and multi-user
modes, respectively? (ii)~When do we use single-user or multi-user
modes? (iii)~For a specific multi-user mode, how do we
quantitatively allocate the wireless resources across multiple
mobile users under the BD based transmissions? (iv)~Which $L$ will
be selected for distributed MIMO transmissions in each fading state
to decrease the average BS-usage while satisfying the QoS
requirements?

Clearly, we can not examine all combinations of
$(\Omega_L,\mathcal{U}_L)$ to minimize the BS usage due to the too
high complexity. Then, in Section~\ref{sect-multiple-BD-multi-mode},
we develop the heuristic algorithms to efficiently select
$\mathfrak{M}_L$ for the specified $L$ in multi-user transmission
modes. In Section~\ref{sect-multiple-BD-single-mode}, we determine
how to select $\mathfrak{M}_L$ in single-user transmission mode.
Based on schemes developed in
Sections~\ref{sect-multiple-BD-multi-mode} and
\ref{sect-multiple-BD-single-mode}, we further answer
questions~(iii) and~(iv) through formulating and solving the joint
BD-PT based
 BS-usage minimization problem in
Sections~\ref{sect-multiple-BD-framework-solution}
and~\ref{sect-multiple-BD-solution}.

\subsection{Selection of $\mathfrak{M}_L$ in Multi-User
Transmission Modes} \label{sect-multiple-BD-multi-mode}

In each fading state, we pick $K_{\mathrm{bs}}$ multi-user
transmission modes as candidates for distributed MIMO transmissions.
These $K_{\mathrm{bs}}$ transmission modes corresponds to
$L=1,2,\ldots,K_{\mathrm{bs}}$, respectively, representing different
levels of BS usages. As mentioned previously, the derivation of
global optimal selection strategy in terms of minimizing the average
BS usage is intractable, since the complexity of examining all
possible $\mathfrak{M}_L=(\Omega_L,\mathcal{U}_L)$ is too high.
Therefore, for a given $L$, we determine $\mathfrak{M}_L$ through a
two-step method. We first propose the priority BS-selection to
determine the BS subset $\Omega_L$. Then, based on the selected
$\Omega_l$, we derive $\mathcal{U}_L$ through a joint
channel-priority user-selection process. \vspace{3pt}

\noindent\emph{A.1.~Priority BS-Selection to Determine $\Omega_L$}


Consider any fading state $\mathbf{H}$. The $n$th user's global
maximum achievable transmission rate is attained when all BS's are
used and all the other users do not transmit. In this case, we have
$L=K_{\mathrm{bs}}$ and
$\mathbf{H}_{\Omega_L}^{(n)}\!\!\!=\mathbf{H}_n$. Moreover, all BS's
and the $n$th user builds a single-user MIMO channel $\mathbf{H}_n$.
Then, the maximum achievable rate is equal to the capacity for the
MIMO channel $\mathbf{H}_n$ with power $\mathcal{P}_L$, which is
given by~\cite{E-Telatar}
\begin{eqnarray}
R^{(n)}_{\max}&\!\!\!\!\!=&\!\!\!\!\max_{\boldsymbol{\Xi}^{(n)}:
\mathrm{Tr}\left(\boldsymbol{\Xi}^{(n)}\right)=\mathcal{P}_{K_{\mathrm{bs}}}}
\!\!\!\left\{\!BT\!\log\left[\det\left(\mathbf{I} +
\mathbf{H}_n\boldsymbol{\Xi}^{(n)}\mathbf{H}_n^{\dag}\right)\right]\right\}\nonumber
\end{eqnarray}
where 
$\det(\cdot)$ generates the determinant of a matrix,
$\mathrm{Tr}(\cdot)$ evaluates the trace of a matrix, and
$\boldsymbol{\Xi}^{(n)}$ is the covariance matrix of
$\mathbf{s}_{\mathfrak{M}_L}^{(n)}$. Correspondingly, we get the
maximum achievable effective capacity of the $n$th user, denoted by
$\mathcal{C}_{\max}^{(n)}$, as follows:
\begin{eqnarray}
\mathcal{C}_{\max}^{(n)} =
-\frac{1}{\theta_n}\log\left(\mathbb{E}_{\mathbf{H}}\left\{e^{-\theta_n
R^{(n)}_{\max}}\right\}\right),
\end{eqnarray}
for $n = 1,2,\ldots,K_{\mathrm{mu}}$. Furthermore, we define the
effective-capacity fraction for the $n$th user as the ratio between
the traffic loads and the maximum achievable effective capacity.
Denoting the effective-capacity fraction by $\widehat{C}_n$, we
define
$\widehat{C}_n\triangleq\overline{C}_n/\mathcal{C}_{\max}^{(n)}$.
Note that $\widehat{C}_n$ can be readily obtained off-line based on
the statistical information of wireless channels, and thus can be
used to design the BS selection algorithm during the data
transmission process. For presentation convenience, we sort
$\{\widehat{C}_n\}_{n=1}^{K_{\mathrm{mu}}}$ in the decreasing order
and denote the permuted version by
$\{\widehat{C}_{\pi(j)}\}_{h=1}^{K_{\mathrm{mu}}}$, where
$\widehat{C}_{\pi(1)}\geq \widehat{C}_{\pi(2)}\geq\cdots\geq
\widehat{C}_{\pi(K_{\mathrm{mu}})}$ indicates the order from the
higher priority to the lower priority. In the rest of this paper, we
use the term of user $\pi(i)$ to denote the user associated with the
$i$th largest effective-capacity fraction.

Clearly, for a higher $\widehat{C}_n$, the $n$th user needs more
wireless resources to meet its QoS requirements. Thus, in order to
satisfy the QoS requirements for all users, we assign higher
BS-selection priority to the user with larger $\widehat{C}_n$.
Following this principle, we design the \emph{priority BS-selection}
algorithm to determine $\Omega_L$ in each fading state and provide
the pseudo code in Fig.~\ref{fig-priority-selection}. In the pseudo
code given by Fig.~\ref{fig-priority-selection}, we use temporary
variables $\overline{\Psi}$ and $\Psi$ to denote the subsets of BS's
which have been selected and which have not been selected,
respectively.

\begin{figure}
\footnotesize \vspace{5pt} \centerline{
\begin{tabular}{p{8.3cm}}
\hline\\
\end{tabular}
}\vspace{-6pt} \centerline{
\begin{tabular}{p{0.05cm} p{7.95cm}}
01. & Let $\Psi :=\{1,2,\ldots,K_{\mathrm{bs}}\}$,
$\overline{\Psi}:= \varnothing$, and $\ell = |\overline{\Psi}|$; ! Initialization \\
02. & $j:=1$. \,\quad\quad\quad\quad! Start selection with User
$\pi(1)$
\\
03. & WHILE $(\ell<L)$ \quad ! Iterative selections until $L$ BS's
are selected
\\
04. & ~~$m^* = \arg\min_{m\in \Psi}\{\gamma_{\pi(j),m}\}$. \\
& ~~~~! User $\pi(j)$ selects the BS with the largest aggregate
power gain.
\\
05. & ~~$\overline{\Psi} := \overline{\Psi} \cup \{m^*\}$,
$\Psi:=\Psi - \{m^*\}$, and $\ell:=\ell+1$.\\
& ~~~~! Update $\overline{\Psi}$, $\Psi$, and $\ell$.
\\
06. & ~~IF $j=K_{\mathrm{mu}}$, then $j:=1$; ELSE $j:=j+1$. \\
& ~~~~! Let next user with lower priority to select BS.
\\
07. & END\\
08. & $\Omega_L:= \overline{\Psi}$. \quad\quad\quad! Complete the BS
selection and get $\Omega_L$.
\\
\hline
\end{tabular}
} \caption{The pseudo codes to determine $\Omega_L$ in each fading
state by using the priority BS-selection algorithm for the
multi-user transmissions.} \label{fig-priority-selection}
\vspace{-15pt}
\end{figure}


As shown in Fig.~\ref{fig-priority-selection}, in each fading state
the BS-selection procedure starts with the selection for user
$\pi(1)$, who has the highest priority. After picking one BS for
user $\pi(1)$, we select one different BS for user $\pi(2)$. More
generally, after selecting for user $\pi(j)$, we choose one BS for
user $\pi(j+1)$ from the BS-subset $\Psi$, which consists of the
BS's that have not been selected. This procedure repeats until $L$
BS's are selected. For user-$\pi(j)$'s selection, we choose the BS
with the maximum aggregate power gain over the subset $\Psi$, where
$\gamma_{\pi(j),m}$ denotes the instantaneous aggregate power gain
between user $\pi(j)$ and the $m$th BS (see
Eq.~(\ref{eq-aggregate-power-gain}) for its definition). In
addition, after user-$\pi(K_{\mathrm{mu}})$'s selection, if the
number of selected BS's is still smaller than $L$, we continue
selecting one more BS for user $\pi(1)$, as shown in line~06 in
Fig.~\ref{fig-priority-selection}, and repeat this iterative
selection procedure until having selected $L$ BS's. Clearly, users
with higher priorities benefitted more from the above algorithm.
Also note that the mobile users' priority order is determined by the
effective-capacity fraction, which adapts to the mobile users' QoS
requirements.\vspace{3pt}


%

\noindent\emph{A.2.~The Principle of the Block Diagonalization
Technique}


The block-diagonalization (BD) precoding
techniques~\cite{Q-H-Spencer} have been widely used for MIMO
transmissions because of its low complexity. In this section, we
also apply the BD technique for our QoS-aware BS selection
framework. For completeness of this paper, the principles of the BD
technique are summarized as follows.


Given transmission mode $\mathfrak{M}_L=(\Omega_L,\mathcal{U}_L)$,
the idea of block diagonalization~\cite{Q-H-Spencer} is to use a
precoding matrix, denoted by
$\boldsymbol{\Gamma}_{\mathfrak{M}_L}^{(n)}$, for the $n$th user's
transmitted signal vector, where $n=n_u\in\mathcal{U}_L$ for some
$u$, such that
$\mathbf{H}_{\Omega_L}^{(i)}\boldsymbol{\Gamma}_{\mathfrak{M}_L}^{(n)}
=\mathbf{0}$ for all $i$ satisfying $i\neq n$ and
$i\in\mathcal{U}_L$. By setting $\mathbf{s}_{\mathfrak{M}_L}^{(n)}=
\boldsymbol{\Gamma}_{\mathfrak{M}_L}^{(n)}\widehat{\mathbf{s}}_{\mathfrak{M}_L}^{(n)}$,
where $\widehat{\mathbf{s}}_{\mathfrak{M}_L}^{(n)}$ is the $n$th
user's data vector to be precoded by
$\boldsymbol{\Gamma}_{\mathfrak{M}_L}^{(n)}$, we can rewrite the
received signal $\mathbf{y}_{\mathfrak{M}_L}^{(n)}$ as
\begin{eqnarray}
\begin{array}{l}
\mathbf{y}_{\mathfrak{M}_L}^{(n)} =\mathbf{H}_{\Omega_L}^{(n)}
\sum_{i\in \mathcal{U}_L}
\boldsymbol{\Gamma}_{\mathfrak{M}_L}^{(i)}\widehat{\mathbf{s}}_{\mathfrak{M}_L}^{(i)}+\boldsymbol{\varsigma}^{(n)}
=
\widehat{\boldsymbol{\Gamma}}_{\mathfrak{M}_L}^{(n)}\widehat{\mathbf{s}}_{\mathfrak{M}_L}^{(n)}
+ \boldsymbol{\varsigma}^{(n)},
\end{array}\nonumber
\end{eqnarray}
where
$\widehat{\boldsymbol{\Gamma}}_{\mathfrak{M}_L}^{(n)}\triangleq\mathbf{H}_{\Omega_L}^{(n)}
\boldsymbol{\Gamma}^{(n)}_{\mathfrak{M}_L}$. Under this strategy,
the $n$th user's signal will not cause interferences to other active
users. Accordingly, the MIMO broadcast transmissions are virtually
converted to $U$ orthogonal MIMO channels with channel matrices
$\big\{\widehat{\boldsymbol{\Gamma}}_{\mathfrak{M}_L}^{(n)}\big\}_{n\in\mathcal{U}_L}$.
Thus, the $n$th user's maximum achievable rate, denoted by
$R^{(n)}_{\mathfrak{M}_L}\big(\mathcal{P}_L^{(n)}\big)$, is equal to
the capacity of the equivalent MIMO channel
$\widehat{\boldsymbol{\Gamma}}_{\mathfrak{M}_L}^{(n)}$, as follows:
\begin{eqnarray}
&&\hspace{-25pt}
~R^{(n)}_{\mathfrak{M}_L}\!\big(\mathcal{P}_L^{(n)}\big)
\!\triangleq\max_{\boldsymbol{\Xi}^{(n)}}
\!\left\{\!BT\log\!\left[\det\!\!\left(\!\mathbf{I} +
\widehat{\boldsymbol{\Gamma}}_{\mathfrak{M}_L}^{(n)}\boldsymbol{\Xi}^{(n)}\!\!\left(\widehat{\boldsymbol{\Gamma}}_{\mathfrak{M}_L}^{(n)}\right)^{\!\!\dag}\right)\!\right]\!\right\}
\nonumber
\\
&&\hspace{-40pt}
\label{eq-multiple-Rn-definition}
\end{eqnarray}
subject to
$\mathrm{Tr}\big(\boldsymbol{\Xi}^{(n)}\big)=\mathcal{P}_L^{(n)}$
for $n\in\mathcal{U}_L$, where $\boldsymbol{\Xi}^{(n)}$ is the
covariance matrix of $\widehat{\mathbf{s}}_{\mathfrak{M}_L}^{(n)}$
and $\mathcal{P}_L^{(n)}$ denotes the power allocated for the $n$th
user under mode $\mathfrak{M}_L$. Correspondingly, we will set the
service rate $R_n$ of the $n$th user equal to
$R^{(n)}_{\mathfrak{M}_L}\big(\mathcal{P}_L^{(n)}\big)$. Note that
$\boldsymbol{\Gamma}_{\mathfrak{M}_L}^{(n)}$ may not exist, which
then results in a service rate equal to 0. Also, we set
$R_n=R^{(n)}_{\mathfrak{M}_L}\big(\mathcal{P}_L^{(n)}\big)= 0$ for
$n\notin\mathcal{U}_L$ or $L=0$. For the procedures to determine the
precoding matrix $\boldsymbol{\Gamma}_{\mathfrak{M}_L}^{(n)}$ of the
$n$th user, where $n=n_u\in\mathcal{U}_L$ for some $u$, please refer
to~\cite{Q-H-Spencer}.

\noindent\emph{A.3.~Derivation of Active-User Set $\mathcal{U}_L$}

Note that given $\Omega_L$ we may not be able to
accommodate all users, because of the limited number transmit
antennas. Although several algorithms for selecting active-user set
have been proposed~\cite{T-Yoo,S-Kaviani}, they cannot be applied in
the framework of this paper, because the QoS provisioning for mobile
users are not addressed those in these algorithms. Next, we
determine $\mathcal{U}_L$ through a joint channel-priority method
for active user selections. The pseudo code of this algorithm is
provided in Fig.~\ref{fig-user-selection}.

\begin{figure}
\footnotesize \vspace{7pt} \centerline{
\begin{tabular}{p{8.3cm}}
\hline\\
\end{tabular}
}\vspace{-6pt} \centerline{
\begin{tabular}{p{0.05cm} p{7.95cm}}
01. & Let $\Lambda :=\{1,2,\ldots,K_{\mathrm{mu}}\}$,
$\overline{\Lambda}:= \varnothing$, and $M_{\Sigma}\triangleq
\sum_{\ell\in\Omega}M_{\ell}$.
\\
02. & WHILE $(\Lambda\neq\varnothing)$ \quad 
\\
03. & ~~For all $n\in\Lambda$ \\
04. & ~~~~Temporarily setting $\mathcal{U}_L:=\overline{\Lambda}
\cup
\{n\}$. \\
05. & ~~~~Get $\boldsymbol{\Gamma}_{\mathfrak{M}_L}^{(n)}$.
Then, set\vspace{-7pt}
\begin{eqnarray}
~~~~\varpi_n&\!\!\!\!\!\!\!:=&\!\!\!\!\!\!
\frac{1}{M_{\Sigma}}\mathbb{E}_{\mathbf{H}}\left\{\left\|\mathbf{H}_{\Omega_L}^{(n)}\boldsymbol{\Gamma}_{\mathfrak{M}_L}^{(n)}\right\|_F^2\nonumber
\left|\mbox{Fixing}~\boldsymbol{\Gamma}_{\mathfrak{M}_L}^{(n)}\right.\right\}\\
&\!\!\!\!\!\!\!=&\!\!\!\!\!\!
\frac{1}{M_{\Sigma}}\overline{\boldsymbol{\Upsilon}}_{\Omega_L}^{(n)}\left[\mathrm{conj}\left(\boldsymbol{\Gamma}_{\mathfrak{M}_L}^{(n)}\right)\circ\boldsymbol{\Gamma}_{\mathfrak{M}_L}^{(n)}\right],
\nonumber
\end{eqnarray}
\vspace{-10pt}\\
& ~~~~where $\overline{\boldsymbol{\Upsilon}}_{\Omega_L}^{(n)}$ is
given by Eq.~(\ref{eq-multiple-BD-csi-avg-power-gain}); $(\cdot\circ\cdot)$ generates \\
& ~~~~element-wise product between two matrices with the same
size;\\
&~~~~$\mathrm{conj}(\cdot)$ yields the element-wise conjugation.
\\
06. & ~~~~Set\vspace{-5pt}
\begin{eqnarray}
~~~~\widehat{\gamma}_n \!:= \!\left\{\!
\begin{array}{ll}
1, & \mbox{if} ~  0<\varpi_n \leq
\frac{1}{M_{\Sigma}}\left\|\mathbf{H}_{\Omega_L}^{(n)}\boldsymbol{\Gamma}_{\mathfrak{M}_L}^{(n)}
\right\|_F^2;
\\
0, & \mbox{otherwise}.
\end{array}
\right.\nonumber
\end{eqnarray}
\vspace{-10pt}\\
07. & ~~END\\
08. & ~~Select $\widehat{u}$ such that for all $j\in\Lambda$,
$j\neq\widehat{u}$, the following condition\vspace{-5pt}
\begin{eqnarray}
~~~(\widehat{\gamma}_{\widehat{u}}>\widehat{\gamma}_j) ~\mbox{or}~
(\widehat{\gamma}_{\widehat{u}}=\widehat{\gamma}_j~\&~
 \mbox{user} ~ \widehat{u} ~\mbox{has higher priority than user $j$})\nonumber
\end{eqnarray}
\vspace{-13pt}
\\
& ~~holds, where the priority order is determined in
Section~\ref{sect-multiple-BD-multi-mode}.1.
\\
09.& ~~IF $\varpi_{\widehat{u}}>0$,
$\overline{\Lambda}:=\overline{\Lambda}\cup\{\widehat{u}\}$ and
$\Lambda:=\Lambda-\{\widehat{u}\}$; else BREAK.\\
10. & END\\
11. & Set $\mathcal{U}_L:= \overline{\Lambda}$.\\
\hline
\end{tabular}
}  \caption{Pseudo codes of the block-diagonalization based joint
channel-priority algorithm to determine the active-user set
$\mathcal{U}_L$ in each fading state.}\vspace{-15pt}
\label{fig-user-selection}
\end{figure}

In the joint channel-priority algorithm provided by
Fig.~\ref{fig-user-selection}, we iteratively select users one by
one into the set $\mathcal{U}_L$. In particular, we use variables
$\Lambda$ and $\overline{\Lambda}$ to represent the temporary sets
of users which have and have not been selected, respectively. As
shown in Fig.~\ref{fig-user-selection}, lines~02 through~10 describe
loops for iterative user selection, where we pick one user in each
loop until all users are selected (i.e., $\Lambda=\varnothing$) or
no more user can be accommodated (examined by line 12). Within each
loop, given the existing active-user set $\overline{\Lambda}$ we
examine the channel quality of each user after BD. Specifically, we
first get the BD precoding matrix of the $n$th user. Then, we derive
$\varpi_n$, which is average channel-power-gain after BD over all
transmit antennas, representing the average channel quality, and
also obtain
$\|\mathbf{H}_{\Omega_L}^{(n)}\boldsymbol{\Gamma}_{\mathfrak{M}_L}^{(n)}\|_F^2/M_{\Sigma}$
line 06, which characterizes the instantaneous channel quality. We
further define a variable $\widehat{\gamma}_n$, as shown in line~07,
where $\widehat{\gamma}_n=1$ and $\widehat{\gamma}_n=0$ indicate
that the channel quality is above and below the average level,
respectively. Obtaining $\widehat{\gamma}_n$, our selection criteria
are as follows. First, we desire to select the user with higher
$\widehat{\gamma}_n$, implying that this user's current channel is
better compared with its statistical channel qualities, which will
more efficiently use the system resources towards this user's QoS
requirement. Second, if two users have the same
$\widehat{\gamma}_n$, we will select the user with higher priority.
Following this criterion, in line 08, we pick the unique user from
$\Lambda$ in the current loop, whose index is denoted $\widehat{u}$.
However, if $\varpi_{\widehat{u}}=0$, implying the maximum
achievable rate equal 0. As a result, no more user can be admitted,
including the $\widehat{u}$th user. We will then terminate the loop,
as shown in line 09, to finish the selection process.


\subsection{BS Selection in Single-User Transmission Modes}
\label{sect-multiple-BD-single-mode}

For single-user transmission modes, we have $\mathcal{U}_L=\{n\}$,
$n\in\{1,2,\ldots,K_{\mathrm{mu}}\}$. Thus, at any time instant,
there is only one user receiving data from multiple BS's through a
single-user MIMO channel $H_{\Omega_L}^{(n)}$. Accordingly, the
maximum achievable rate for the $n$th user is equal to the capacity
of $H_{\Omega_L}^{(n)}$ with power $\mathcal{P}_L$, which is denoted
by $R_{\Omega_L}^{(n)}$. However, even for the single-user case, the
complexity of high of choosing $\Omega_L$ to maximize the achievable
data rate is too high, since we need to examine all $K_{\mathrm{bs}}
\choose L$ combinations. Norm-based antenna selections have been
demonstrated to be effective in achieving high system throughput
with low complexity in centralized MIMO
system~\cite{S-Sanayei,M-Gharavi-Alkhansari}, which can be also
extended to BS-selection in distributed MIMO system. Specifically,
for the $n$th user with specified $L$ in our framework, we select
BS's with $L$ largest aggregate channel power gain. As a result, in
each fading state we have $K_{\mathrm{bs}}K_{\mathrm{mu}}$
single-user modes as candidates for distributed MIMO transmissions.
Given the transmission mode with $\Omega_L$ and the active user $n$,
we will set the service rate $R_n$ equal to $R_{\Omega_L}^{(n)}$.


\subsection{The Optimization Framework for Transmission Mode Selection and Resource Allocation}
\label{sect-multiple-BD-framework-solution}

We have derived candidate $(\Omega_L, \mathcal{U}_L)$ in multi-user
and single-user transmissions modes, respectively. We still need to
answer how to allocate power over different mobile users in
multi-user modes and which transmission mode will be eventually used
for distributed MIMO transmissions. In this section, we employ the
probabilistic transmission to determine finally selecting which
transmission mode. Specifically, we use multi-user mode
$(\Omega_L,\mathcal{U}_L)$ determined through algorithms given in
Figs.~\ref{fig-priority-selection} and~\ref{fig-user-selection} with
a probability denoted by $\phi_L$, $L=0,1,2,\ldots,K_{\mathrm{bs}}$;
also, we use single-user mode with BS-subset cardinality $L$ and
$\mathcal{U}_L$ with a probability denoted by $q_{L,n}$,
$L=1,2,\ldots,K_{\mathrm{bs}}$, $n=1,2,\ldots,K_{\mathrm{mu}}$. Note
that $\phi_0$ is the probability of the case that nothing is
transmitted. Clearly, the sum over all $q_{L,n}$ and $\phi_L$ must
be equal to 1. For multi-user mode, we denote power allocated to the
$n$th user in transmission mode $(\Omega_L,\mathcal{U}_L)$ by
$\mathcal{P}_L^{(n)}$, while the total power constraint is given by
Eq.~(\ref{eq-power-linear-increase}). For presentation convenience,
we further define
$\boldsymbol{\phi}\triangleq(\phi_1,\phi_2,\ldots,\phi_{K_{\mathrm{mu}}})$
and
$\mathbf{q}\triangleq(\mathbf{q}_1,\mathbf{q}_2,\ldots,\mathbf{q}_{K_{\mathrm{bs}}})$
with
$\mathbf{q}_L\triangleq(q_{L,1},q_{L,2},\ldots,q_{L,K_{\mathrm{mu}}})$
to describe the probabilistic transmission policy; we also define
$\boldsymbol{\mathcal{P}} \triangleq
\big(\boldsymbol{\mathcal{P}}_1, \boldsymbol{\mathcal{P}}_2, \ldots,
\boldsymbol{\mathcal{P}}_{K_{\mathrm{bs}}} \big)$ with
$\boldsymbol{\mathcal{P}}_L \triangleq \big(\mathcal{P}_{L}^{(1)},
\mathcal{P}_{L}^{(2)},\ldots, \mathcal{P}_{L}^{(K_{\mathrm{mu}})}
\big)$ to characterize the power allocation policy in
\emph{multi-user} modes.
%
%
Then, we formulate the following optimization problem
$\boldsymbol{A1}$ to derive the efficient transmission-mode
selection and the corresponding power allocation policy:
\vspace{3pt}

\noindent $\boldsymbol{A1}$: Joint BD-PT based BS-usage
minimization\vspace{-3pt}
\begin{eqnarray}
&&
\hspace{-50pt}\min_{(\boldsymbol{\phi},\mathbf{q},\boldsymbol{\mathcal{P}})}
\left\{\mathbb{E}_{\mathbf{H}}
\left\{\sum_{L=1}^{K_{\mathrm{bs}}}L\left(\phi_L+\sum_{n=1}^{K_{\mathrm{mu}}}q_{L,n}\right)\right\}\right\}\nonumber\\
&&\hspace{-50pt}\mbox{s.t.: }
\,1).~\sum_{L=0}^{K_{\mathrm{bs}}}\phi_L
+\sum_{L=1}^{K_{\mathrm{bs}}}\sum_{n=1}^{K_{\mathrm{mu}}}q_{L,n}= 1,
\quad\quad \forall\,\mathbf{H}\label{eq-multiple-PT-constraint-pt}\\
&&\hspace{-50pt}2).~\sum_{n=1}^{K_{\mathrm{mu}}}
\mathcal{P}_L^{(n)}= \mathcal{P}_L,\quad\quad \forall\,
L,\,\mathbf{H}; \label{eq-multiple-PT-constraint-power}
\\
&&\hspace{-50pt}3).~\mathbb{E}_{\mathbf{H}}\!\left\{\sum_{L=0}^{K_{\mathrm{bs}}}
\!\left(\!\phi_Le^{-\theta_n
R^{(n)}_{\mathfrak{M}_L}\big(\mathcal{P}_L^{(n)}\big)}+q_{L,n}
e^{-\theta_n R^{(n)}_{\Omega_L}}\!\right)\right.
\nonumber\\
&&\hspace{-50pt}\hspace{-15pt}\quad\quad\quad\quad\quad\quad+\sum_{L=0}^{K_{\mathrm{bs}}}\sum_{j,j\neq
n}q_{L,j}\Bigg\}- e^{-\theta_n\overline{C}_n}\leq 0, ~\forall\,
n.\!\!\!\! \!\!\!\! \!\!\!\! \label{eq-multiple-PT-constraint-effcp}
\end{eqnarray}

\subsection{Derivations of the Optimal Solution of Problem $\boldsymbol{A1}$}
\label{sect-multiple-BD-solution}

 \noindent\emph{D.1.~The Properties of
$R^{(n)}_{\mathfrak{M}_L}\big(\mathcal{P}_L^{(n)}\big)$}

Before solving $\boldsymbol{A1}$, we first summarize the properties
of $R^{(n)}_{\mathfrak{M}_L}\big(\mathcal{P}_L^{(n)}\big)$
determined by Eq.~(\ref{eq-multiple-Rn-definition}). Based on
results in~\cite{E-Telatar}, the $n$th user's MIMO channel
$\widehat{\boldsymbol{\Gamma}}_{\mathfrak{M}_L}^{(n)}$ (after BD)
can be converted to $Z_L^{(n)}$ parallel Gaussian sub-channels,
where $Z_L^{(n)}$ is the rank of
$\widehat{\boldsymbol{\Gamma}}_{\mathfrak{M}_L}^{(n)}$.
Correspondingly, the $z$th sub-channel's SNR is equal to
$\varepsilon^{(n)}_{L,z}$, where the square root of
$\varepsilon^{(n)}_{L,z}$ is the $z$th largest nonzero singular
value of $\widehat{\boldsymbol{\Gamma}}_{\mathfrak{M}_L}^{(n)}$. The
optimal power $\rho^{(n)}_{L,z}$ allocated to the $z$th sub-channel
follows the water-filling allocation, which is equal to
$\rho^{(n)}_{L,z}=\big[\mu_L^{(n)}-1/\varepsilon^{(n)}_{L,z}\big]^+$,
where $[\cdot]^+\triangleq\max\{\cdot,0\}$ and $\mu_L^{(n)}$ is
selected such that
$\sum_{z=1}^{Z_L^{(n)}}\rho^{(n)}_{L,z}=\mathcal{P}_L^{(n)}$. Since
$\widehat{\boldsymbol{\Gamma}}_{\mathfrak{M}_L}^{(n)}$ has only
$Z_L^{(n)}$ non-zero singular values, we define
$1/\varepsilon_{L,i}^{(n)}\triangleq \infty$ for $i=Z_L^{(n)}+1$ and
$1/\varepsilon_{L,i}^{(n)}\triangleq 0$ for $i=0$. We can further
show that $R^{(n)}_{\mathfrak{M}_L}\big(\mathcal{P}_L^{(n)}\big)$ is
a strictly concave function and
\begin{eqnarray}
\begin{array}{l}
\frac{d R^{(n)}_{\mathfrak{M}_L}\big(\mathcal{P}_L^{(n)}\big)}{d
\mathcal{P}_L^{(n)}}   = \frac{BT}{\mu_L^{(n)}}
\end{array}\label{eq-multiple-PT-derivative-R/P}
\end{eqnarray}
holds. Moreover, if
$\mu_L^{(n)}\in\Big[1/\varepsilon_{L,i}^{(n)},1/\varepsilon_{L,i+1}^{(n)}\Big)$
for $i=1,2,\ldots,Z_L^{(n)}$, we get:
\begin{eqnarray}
\hspace{-20pt}\mathcal{P}_L^{(n)} &\!\!\!\!=&\!\!\!\!\!\!
\begin{array}{l}
\left[i \mu_L^{(n)} - \sum_{j=1}^i
\frac{1}{\varepsilon_{L,j}^{(n)}}\right]^+;
\end{array}
\label{eq-multiple-Rn-water-filling-1}
\\
\hspace{-20pt}R^{(n)}_{\mathfrak{M}_L}\big(\mathcal{P}_L^{(n)}\big)
&\!\!\!\!=&\!\!\!\!\!\!
\begin{array}{l}
BT\log\left(\prod_{j=1}^i \varepsilon^{(n)}_{L,j}\right)
+BTi\log\mu_L^{(n)}.
\end{array}
\label{eq-multiple-Rn-water-filling-2}
\end{eqnarray}

\vspace{3pt}\noindent\emph{D.2.~The Optimal Solution to
$\boldsymbol{A1}$}


\begin{theorem}
The optimal solution for optimization problem $\boldsymbol{A1}$, if
existing, is given by
\begin{eqnarray}
\left(\mathcal{P}_L^{(n)}\right)^* = \left\{
\begin{array}{ll}
\left[i^* \mu_L^{(n)} - \sum_{j=1}^{i^*}
\frac{1}{\varepsilon_{L,j}^{(n)}}\right]^+, & \mbox{if} ~
n\in\mathcal{U}_L;\vspace{3pt}\\
0 , & \mbox{if} ~ n\notin\mathcal{U}_L;
\end{array}
\right.
\label{eq-multiple-PT-A1-opt-power}
\end{eqnarray}
for all $n$, $L$, and $\mathbf{H}$, where $\varepsilon^{(n)}_{L,j}$
is the square of
$\widehat{\boldsymbol{\Gamma}}_{\mathfrak{M}_L}^{(n)}$'s $j$th
largest singular value; $(\mu_L^{(n)}, i^*)$ is the unique solution
satisfying the following conditions:
\begin{eqnarray}
&&\hspace{-35pt}\left\{\!\!\!
\begin{array}{lcl}
\mu_L^{(n)} &\!\!\!\!=&
\!\!\!\!\left(\frac{\zeta^*_{\mathbf{H},L}}{BT\theta_n\lambda^*_n}\right)^{-\frac{1}{1+i^*BT\theta_n}}\prod_{j=1}^{i^*}
\left(\varepsilon^{(n)}_{L,j}\right)^{-\frac{BT\theta_n}{1+i^*BT\theta_n}}
;\vspace{3pt}\\
\mu_L^{(n)}
&\!\!\!\!\in&\!\!\!\!\left[\frac{1}{\varepsilon_{L,i^*}^{(n)}},\frac{1}{\varepsilon_{L,i^*+1}^{(n)}}\right),\quad\forall\,n,\,L,\,\mathbf{H}.
\end{array}
\right. \label{eq-multiple-PT-A1-opt-istar-and-mu}
\end{eqnarray}
The corresponding optimal PT policy is determined by
\begin{eqnarray}
\left\{
\begin{array}{lcl}
\phi_L^* = 1_{(\psi_L=\psi^*)};\\
q_{L,n}^* = 1_{(\psi_{L,n}=\psi^*)},
\end{array}
\right.\label{eq-multiple-PT-A1-opt-phi-q}
\end{eqnarray}
where $1_{(\cdot)}$ is the indication function and $\psi^*$ is
defined as
\begin{eqnarray}
\psi^* \triangleq \min\left\{ \min_{L}\left\{\psi_L\right\},
\min_{(L,n)}\left\{\psi_{L,n}\right\}\right\}\label{eq-multiple-PT-A1-opt-psi-star}
\end{eqnarray}
with
\begin{eqnarray}
&&\hspace{-20pt}\left\{
\begin{array}{ccl}
\!\!\!\psi_L &\!\!\!\!\!\triangleq& \!\!\!\!L \!+\!
\sum_{n=1}^{K_{\mathrm{mu}}} \lambda_n^*
e^{-\theta_nR^{(n)}_{\mathfrak{M}_L}\left(\big(\mathcal{P}_L^{(n)}\big)^{\!*}\right)},
~~ 0\leq L\leq K_{\mathrm{bs}};\vspace{3pt}
\\
\!\!\!\psi_{L,n} & \!\!\!\!\!\triangleq & \!\!\!\!L \!+\!
\lambda_n^* e^{-\theta_nR^{(n)}_{\Omega_L}} + \sum_{j,j\neq
n}\lambda_j^*, ~~ L\in[1,K_{\mathrm{bs}}], ~\forall
\,n;
\end{array}
\right. \nonumber
\end{eqnarray}
if multiple $\psi_L$'s and/or $\psi_{L,n}$'s all equal to $\psi^*$,
the corresponding transmission modes will be allocated equal
probability with the sum probability equal to 1. The variables
$\{\lambda_n^*\}_{n=1}^{K_{\mathrm{mu}}}$ are constants over all
fading state; given $\{\lambda_n^*\}_{n=1}^{K_{\mathrm{mu}}}$,
$\zeta^*_{\mathbf{H},L}$ is selected to satisfy the equation
$\sum_{n=1}^{K_{\mathrm{mu}}}\big(\mathcal{P}_L^{(n)}\big)^*=\mathcal{P}_L$
for all $L$ and $\mathbf{H}$; accordingly
$\{\lambda_n^*\}_{n=1}^{K_{\mathrm{mu}}}$ need to be selected such
that the equality of Eq.~(\ref{eq-multiple-PT-constraint-effcp})
holds. \label{theorem-PBS-BD-PT}
\end{theorem}

\begin{proof}
We construct $\boldsymbol{A1}$'s Lagrangian function, denoted by
$\mathcal{J}_{A1}(\boldsymbol{\phi},\mathbf{q},\boldsymbol{\mathcal{P}};\boldsymbol{\lambda},\boldsymbol{\zeta}_{\mathbf{H}})$,
as
\begin{eqnarray}
\mathcal{J}_{A1}(\boldsymbol{\phi},\mathbf{q},\boldsymbol{\mathcal{P}};\boldsymbol{\lambda},\boldsymbol{\zeta}_{\mathbf{H}})
=\mathbb{E}_{\mathbf{H}}\left\{J_{A1}(\boldsymbol{\phi},\mathbf{q},\boldsymbol{\mathcal{P}};\boldsymbol{\lambda},\boldsymbol{\zeta}_{\mathbf{H}})\right\}
\label{eq-multiple-PT-A1-Lagrangian-1}
\end{eqnarray}
subject to
$\sum_{L=0}^{K_{\mathrm{bs}}}\phi_L+\sum_{L=1}^{K_{\mathrm{bs}}}\sum_{n=1}^{K_{\mathrm{mu}}}q_{L,n}=1$,
where
\begin{eqnarray}
&&
\hspace{-25pt}J_{A1}(\boldsymbol{\phi},\mathbf{q},\boldsymbol{\mathcal{P}};\boldsymbol{\lambda},\boldsymbol{\zeta}_{\mathbf{H}})\nonumber
\\
&&\hspace{-20pt}\triangleq
\!\!\sum_{L=0}^{K_{\mathrm{bs}}}L\left(\phi_L+\sum_{n=1}^{K_{\mathrm{mu}}}q_{L,n}\right)
+\sum_{L=1}^{K_{\mathrm{bs}}}\zeta_{\mathbf{H},L}\left(\sum_{n=1}^{K_{\mathrm{mu}}}
\mathcal{P}_{L}^{(n)}- \mathcal{P}_L\right)
\nonumber\\
&&\hspace{-20pt}
+\!\!\sum_{n=1}^{K_{\mathrm{mu}}}\!\lambda_n\Bigg[\sum_{L=0}^{K_{\mathrm{bs}}}\!
\left(\!\phi_Le^{-\theta_n
R^{(n)}_{\Omega_L\!,\,\mathcal{U}_L}\!\big(\!\mathcal{P}_L^{(n)}\!\big)}
\!\!+\!q_{L,n} e^{-\theta_n
R^{(n)}_{\Omega_L}}\!\right)\nonumber\\
&&\quad\quad\quad\quad\quad\quad\quad
+\Bigg(\sum_{L=0}^{K_{\mathrm{bs}}}\sum_{j,j\neq
n}q_{L,j}\Bigg)-\!e^{-\theta_n\overline{C}_n}\!\Bigg].
\label{eq-multiple-PT-A1-Lagrangian-2}
\end{eqnarray}
In
Eqs.~(\ref{eq-multiple-PT-A1-Lagrangian-1})-(\ref{eq-multiple-PT-A1-Lagrangian-2}),
$\boldsymbol{\lambda}\triangleq(\lambda_1,\lambda_2,\ldots,\lambda_{K_{\mathrm{mu}}})$
and $\lambda_n$'s are the Lagrangian multipliers associated with
Eq.~(\ref{eq-multiple-PT-constraint-effcp}), which are constants
over all fading states and satisfies $\lambda_n\geq0$;
$\{\zeta_{\mathbf{H},L}\}_{L=1}^{K_{\mathrm{bs}}}$ are the
Lagrangian multipliers associated with
Eq.~(\ref{eq-multiple-PT-constraint-power}) in each fading state,
and
$\boldsymbol{\zeta}_{\mathbf{H},L}\triangleq(\zeta_{\mathbf{H},1},\zeta_{\mathbf{H},2},\ldots,\zeta_{\mathbf{H},K_{\mathrm{bs}}})$.

The optimization problem $\boldsymbol{A1}$'s Lagrangian dual
function~\cite{M-Bazaraa}, denoted by
$\mathfrak{J}_{A1}(\boldsymbol{\lambda},\boldsymbol{\zeta}_{\mathbf{H}})$,
is determined by
\begin{eqnarray}
\hspace{-20pt}\mathfrak{J}_{A1}(\boldsymbol{\lambda},\boldsymbol{\zeta}_{\mathbf{H}})
&\!\!\!\triangleq&\!\!\!\min_{(\boldsymbol{\phi},\mathbf{q},\boldsymbol{\mathcal{P}})}
\Big\{\mathcal{J}_{A1}(\boldsymbol{\phi},\mathbf{q},\boldsymbol{\mathcal{P}};\boldsymbol{\lambda},\boldsymbol{\zeta}_{\mathbf{H}})
\Big\}
\nonumber\\
&\!\!\!=&\!\!\!
\mathbb{E}_{\mathbf{H}}\left\{\min_{(\boldsymbol{\phi},\mathbf{q},\boldsymbol{\mathcal{P}})}\Big\{
J_{A1}(\boldsymbol{\phi},\mathbf{q},\boldsymbol{\mathcal{P}};\boldsymbol{\lambda},\boldsymbol{\zeta}_{\mathbf{H}})\Big\}\right\}.
\label{eq-multiple-PT-A1-dual-derive-1}
\end{eqnarray}
subject to
$\sum_{L=0}^{K_{\mathrm{bs}}}\phi_L+\sum_{L=1}^{K_{\mathrm{bs}}}\sum_{n=1}^{K_{\mathrm{mu}}}q_{L,n}=1$
for all $\mathbf{H}$. Lagrangian duality theory
shows~\cite{M-Bazaraa} that
$\mathfrak{J}_{A1}(\boldsymbol{\lambda},\boldsymbol{\zeta}_{\mathbf{H}})$
is always a concave function, whose maximizer is upper-bounded by
the optimum of $\boldsymbol{A1}$. We then denote the maximizer of
$\mathfrak{J}_{A1}(\boldsymbol{\lambda},\boldsymbol{\zeta}_{\mathbf{H}})$
by $(\boldsymbol{\lambda}^*,\boldsymbol{\zeta}_{\mathbf{H}}^*)$.
Also, we denote by
$(\boldsymbol{\phi}^*,\mathbf{q}^*,\boldsymbol{\mathcal{P}}^*)$ the
minimizer to Eq.~(\ref{eq-multiple-PT-A1-dual-derive-1}), which
varies with $(\boldsymbol{\lambda},\boldsymbol{\zeta})$. Then, given
$(\boldsymbol{\lambda}^*,\boldsymbol{\zeta}_{\mathbf{H}}^*)$, we
have
\begin{eqnarray}
\hspace{-20pt}(\boldsymbol{\phi}^*,\mathbf{q}^*) &\!\!\!\!=&\!\!\!\!
\arg\min_{(\boldsymbol{\phi},\mathbf{q})}\left\{J_{A1}(\boldsymbol{\phi},\mathbf{q},\boldsymbol{\mathcal{P}}^*;\boldsymbol{\lambda}^*,\boldsymbol{\zeta}_{\mathbf{H}}^*)\right\}
\label{eq-multiple-PT-A1-dual-derive-2-a}\nonumber
\\
& \!\!\!\!\stackrel{(a)}{=} &
\!\!\!\!\arg\min_{(\boldsymbol{\phi},\mathbf{q})}\left\{\sum_{L=0}^{K_{\mathrm{bs}}}
\phi_L\psi_L +
\sum_{L=1}^{K_{\mathrm{bs}}}\sum_{n=1}^{K_{\mathrm{mu}}}q_{L,n}\psi_{L,n}\right\},
\label{eq-multiple-PT-A1-dual-derive-2}
\end{eqnarray}
for all $\mathbf{H}$, where $\psi_L$ and $\psi_{L,n}$ is defined in
Theorem~\ref{theorem-PBS-BD-PT}, and equation $(a)$ holds by
applying Eq.~(\ref{eq-multiple-PT-A1-Lagrangian-2}) and removing the
terms independent of $\boldsymbol{\phi}$. Solving
Eq.~(\ref{eq-multiple-PT-A1-dual-derive-2}) subject to
$\sum_{L=0}^{K_{\mathrm{bs}}}\phi_L+\sum_{L=1}^{K_{\mathrm{bs}}}\sum_{n=1}^{K_{\mathrm{mu}}}q_{L,n}=1$,
we obtain
Eqs.~(\ref{eq-multiple-PT-A1-opt-phi-q})-(\ref{eq-multiple-PT-A1-opt-psi-star}).
If multiple $\psi_L$'s and/or $\psi_{L,n}$'s all equal to $\psi^*$,
which happens with probability zero when $\mathbf{H}$ has continues
CDF, how to allocate probabilities across these modes does not
affect the eventual results. Therefore, without loss of generality
we allocate the corresponding transmission modes with equal
probability while keeping their sum equal to 1.

It is clear that $\big(\mathcal{P}_L^{(n)}\big)^*=0$ for
$n\in\mathcal{U}_L$. Next, we consider $n\in\mathcal{U}_L$. Based on
Eqs.~(\ref{eq-multiple-PT-A1-opt-phi-q})-(\ref{eq-multiple-PT-A1-opt-psi-star}),
the opportunity of transmitting the data in a fading state will be
given to only one transmission mode. Moreover, given $\phi_L=1$ for
some mode $\mathcal{M}_L$, the power allocations for other mode do
not affect the Lagrangian function. Therefore,
$\boldsymbol{\mathcal{P}}_L^*$ needs to minimize
$J_{A1}(\boldsymbol{\phi},\mathbf{q},\boldsymbol{\mathcal{P}};\boldsymbol{\lambda}^*,\boldsymbol{\zeta}_{\mathbf{H}}^*)$
under $\phi_L=1$, $\phi_j=0$ for all $j\neq L$, and $q_{L,n}=0$ for
all $L,n$. We denote
$J_{A1}(\boldsymbol{\phi},\mathbf{q},\boldsymbol{\mathcal{P}};\boldsymbol{\lambda}^*,\boldsymbol{\zeta}_{\mathbf{H}}^*)$
under this condition by
$J_{A1,L}(\boldsymbol{\mathcal{P}};\boldsymbol{\lambda}^*,\boldsymbol{\zeta}_{\mathbf{H},L}^*)$.
Then, applying Eq.~(\ref{eq-multiple-PT-derivative-R/P}), taking the
derivative of
$J_{A1,L}(\boldsymbol{\mathcal{P}};\boldsymbol{\lambda}^*,\boldsymbol{\zeta}_{\mathbf{H},L}^*)$
with respect to (w.r.t.) $\mathcal{P}_L^{(n)}$, and letting the
derivative equal to zero, we get
\begin{eqnarray}
\zeta_{\mathbf{H},L}^*-BT\lambda_n^*\theta_n  \mu_L^{(n)}
e^{-\theta_n
R^{(n)}_{\mathfrak{M}_L}\big(\mathcal{P}_L^{(n)}\big)}=0, \quad
\forall n,~L,~\mathbf{H}.\label{eq-multiple-PT-A1-power-derive-1}
\end{eqnarray}
Deriving $\mu_L^{(n)}$ and applying
Eq.~(\ref{eq-multiple-Rn-water-filling-1}), we obtain
Eqs.~(\ref{eq-multiple-PT-A1-opt-power})-(\ref{eq-multiple-PT-A1-opt-istar-and-mu}).

We further define
\begin{eqnarray}
\left\{\!\!\!\!
\begin{array}{ccl}
f_n(\boldsymbol{\phi},\mathbf{q},\boldsymbol{\mathcal{P}})&\!\!\!\!\triangleq&\!\!\!\!\mathbb{E}_{\mathbf{H}}\Big\{\sum_{L=0}^{K_{\mathrm{bs}}}\!
\Big(\phi_Le^{-\theta_n
R^{(n)}_{\mathfrak{M}_L}\!\big(\!\mathcal{P}_L^{(n)}\!\big)}
\!\!\\
&&\!\!\!\!+q_{L,n} e^{-\theta_n R^{(n)}_{\Omega_L}}\!+\sum_{j,j\neq
n}q_{L,j}\Big)-e^{-\theta_n\overline{C}_n}\!\Big\};\\
f_{\mathbf{H},L}(\boldsymbol{\mathcal{P}}_L) &\!\!\!\!\triangleq&
\!\!\!\!\sum_{n=1}^{K_{\mathrm{mu}}} \mathcal{P}_{L}^{(n)}-
\mathcal{P}_L
\end{array}
\right.\nonumber
\end{eqnarray}
which are the constraint functions on the left-hand sides of
Eqs.~(\ref{eq-multiple-PT-constraint-effcp})
and~(\ref{eq-multiple-PT-constraint-power}), respectively. The
Lagrangian duality principle~\cite{M-Bazaraa} suggests that the
optimal objective value $\overline{L}^*$ of $\boldsymbol{A1}$
satisfies:
\begin{eqnarray}
\overline{L}^*\geq
\mathfrak{J}_{A1}(\boldsymbol{\lambda}^*,\boldsymbol{\zeta}_{\mathbf{H}}^*).
\label{eq-multiple-PT-A1-duality-gap}
\end{eqnarray}
Also, $f_{\mathbf{H},L}(\boldsymbol{\mathcal{P}}_L^*)$ and
$f_n(\boldsymbol{\phi}^*,\mathbf{q}^*,\boldsymbol{\mathcal{P}}^*)$
are the subgradients~\cite{M-Bazaraa} of
$\mathfrak{J}_{A1}(\boldsymbol{\lambda},\boldsymbol{\zeta}_{\mathbf{H}})$
w.r.t. $\zeta_{\mathbf{H},L}$ and $\lambda_n$, respectively. We can
further prove that the subgradients
$f_n(\boldsymbol{\phi}^*,\mathbf{q}^*,\boldsymbol{\mathcal{P}}^*)$
and $f_{\mathbf{H},L}(\boldsymbol{\mathcal{P}}_L^*)$ of
$\mathfrak{J}_{A1}(\boldsymbol{\lambda},\boldsymbol{\zeta}_{\mathbf{H}})$
vary continuously with
$(\boldsymbol{\lambda},\boldsymbol{\zeta}_{\mathbf{H}})$. Thus,
$\mathfrak{J}_{A1}(\boldsymbol{\lambda},\boldsymbol{\zeta}_{\mathbf{H}})$
is differentiable and we have
$\partial\mathfrak{J}_{A1}(\boldsymbol{\lambda},\boldsymbol{\zeta}_{\mathbf{H}})/\partial\lambda_n=f_n(\boldsymbol{\phi}^*,\mathbf{q}^*,\boldsymbol{\mathcal{P}}^*)$
and
$\partial\mathfrak{J}_{A1}(\boldsymbol{\lambda},\boldsymbol{\zeta}_{\mathbf{H}})/\partial\zeta_{\mathbf{H},L}=f_{\mathbf{H},L}(\boldsymbol{\mathcal{P}}_L^*)g(\mathbf{H})d\mathbf{H}$,
where $g(\mathbf{H})$ is the probability density function (pdf) of
$\mathbf{H}$ and $d\mathbf{H}$ denotes the integration variable.

It is clear that if
$f_n(\boldsymbol{\phi}^*,\mathbf{q}^*,\boldsymbol{\mathcal{P}}^*)=0$
(for all $n$) and $f_{\mathbf{H},L}(\boldsymbol{\mathcal{P}}_L^*)=0$
(for all $L$ and $\mathbf{H}$) hold,
$\mathfrak{J}_{A1}(\boldsymbol{\lambda},\boldsymbol{\zeta}_{\mathbf{H}})$
attains its maximum. Since
$f_{\mathbf{H},L}(\boldsymbol{\mathcal{P}}_L^*)$ monotonically
varies with $\zeta_{\mathbf{H},L}$ as observed from
Eq.~(\ref{eq-multiple-PT-A1-opt-power})-(\ref{eq-multiple-PT-A1-opt-istar-and-mu}),
we can show that for any $\boldsymbol{\lambda}$, there exists a
$\zeta_{\mathbf{H},L}'$ resulting in
$f_{\mathbf{H},L}(\boldsymbol{\mathcal{P}}_L^*)=0$,
$L=1,2,\ldots,K_{\mathrm{bs}}$. This implies that
$\zeta_{\mathbf{H},L}^*$ must be selected such that the equality
holds in Eq.~(\ref{eq-multiple-PT-constraint-power}) under
$\boldsymbol{\lambda}^*$. Due to the concavity of
$\mathfrak{J}_{A1}(\boldsymbol{\lambda},\boldsymbol{\zeta}_{\mathbf{H}})$,
$\partial\mathfrak{J}(\boldsymbol{\lambda},\boldsymbol{\zeta}_{\mathbf{H}}')/\partial\lambda_n$
is a decreasing function of $\lambda_n$. Also, we can readily show
that
$\partial\mathfrak{J}(\boldsymbol{\lambda},\boldsymbol{\zeta}_{\mathbf{H}}')/\partial\lambda_n|_{\lambda_n=0}>0$.
Then, if there does not exist $\boldsymbol{\lambda}$ such that
$\partial\mathfrak{J}(\boldsymbol{\lambda},\boldsymbol{\zeta}_{\mathbf{H}}')/\partial\lambda_n=0$
for all $n$, we have $\lambda_n^*\rightarrow\infty$ for some $n$th
user and
$\partial\mathfrak{J}(\boldsymbol{\lambda},\boldsymbol{\zeta}_{\mathbf{H}}^*)/\partial\lambda_n>0$
always holds. For this case, we get
$\overline{L}^*\geq\mathfrak{J}(\boldsymbol{\lambda}^*,\boldsymbol{\zeta}_{\mathbf{H}}^*)
\rightarrow\infty$, implying no feasible solution for
$\boldsymbol{A1}$.

In contrast, if there exists $\boldsymbol{\lambda}^*$ such that
$\partial\mathfrak{J}_{A1}(\boldsymbol{\lambda}^*,\boldsymbol{\zeta}_{\mathbf{H}}^*)/\partial\lambda_n=0$
for all $n$, we have $\zeta_{\mathbf{H},L}^*=\zeta_{\mathbf{H},L}'$
and the obtained
$(\boldsymbol{\phi}^*,\mathbf{q}^*,\boldsymbol{\mathcal{P}}^*)$ is
feasible to $\boldsymbol{A1}$. Moreover, we get
$\overline{L}^*=\mathfrak{J}_{A1}(\boldsymbol{\lambda}^*,\boldsymbol{\zeta}_{\mathbf{H}}^*)$
with zero duality gap~\cite{M-Bazaraa} by examining
Eq.~(\ref{eq-multiple-PT-A1-Lagrangian-1}), implying that
$(\boldsymbol{\phi}^*,\mathbf{q}^*,\boldsymbol{\mathcal{P}}^*)$
given by
Eqs.~(\ref{eq-multiple-PT-A1-opt-power})-(\ref{eq-multiple-PT-A1-opt-psi-star})
under $\boldsymbol{\lambda}^*$ and
$\boldsymbol{\zeta}_{\mathbf{H}}^*$ is optimal solution of
$\boldsymbol{A1}$, and thus Theorem~\ref{theorem-PBS-BD-PT} follows.
\end{proof}

%

Note that there are no closed-form solutions for the optimal
Lagrangian multipliers $\boldsymbol{\zeta}^*_{\mathbf{H}}$ and
$\boldsymbol{\lambda}^*$. In each fading state,
$\zeta^*_{\mathbf{H},L}$ needs to be selected to satisfy
$f_{\mathbf{H},L}(\boldsymbol{\mathcal{P}}_L^*)=0$, as discussed in
the proof of Theorem~\ref{theorem-PBS-BD-PT}, which can be
conveniently determined through numerical searching method in that
$f_{\mathbf{H},L}(\boldsymbol{\mathcal{P}}_L^*)$ varies
monotonically with $\zeta_{\mathbf{H},L}$. Moreover, we can
determine $\boldsymbol{\zeta}^*_{\mathbf{H}}$ through maximizing the
Lagrangian dual function
$\mathfrak{J}(\boldsymbol{\zeta}_{\mathbf{H}},\boldsymbol{\lambda})$
by using the gradient descent algorithm. Due to the concavity of
$\mathfrak{J}(\boldsymbol{\zeta}_{\mathbf{H}},\boldsymbol{\lambda})$,
the gradient descent algorithm will converge with appropriately
selected step size. If the gradient descent algorithm does not
converge with $\lambda_n$ approaching infinity, the optimal solution
does not exist for $\boldsymbol{A1}$, as discussed in the proof of
Theorem~\ref{theorem-PBS-BD-PT}, which implies that the current
wireless resources cannot simultaneously support QoS requirements
for all of current mobile users.

\subsection{Pure PT Based BS-Selection}

We further consider the BS-selection framework based on the PT-only
approach for multiple access across mobile users. In this framework,
the system only considers the $K_{\mathrm{bs}}K_{\mathrm{mu}}$
single-user modes derived in
Section~\ref{sect-multiple-BD-single-mode} and the mode transmitting
nothing as candidates for distributed MIMO transmissions. This
PT-only based framework also uses probabilistic transmission to
determine which transmission mode is used. Then, we can formulate
the corresponding BS-usage minimization problem subject to the same
power and QoS constraints as in problem $\boldsymbol{A1}$, where
only the probability vector assigned for the
$K_{\mathrm{bs}}K_{\mathrm{mu}}+1$ candidate modes can be tuned to
minimize the average BS-usage. The detailed problem descriptions and
the corresponding optimal solution is omitted due to lack of space,
but provided on-line in~\cite{Q-Du-online}. It is clear that this
framework is easier to implement as compared to the joint BD-PT
approach, but it can only support the lower traffic load.

\section{TDMA Based BS-Selection Scheme}
\label{sect-multiple-TDMA}


We next study the TDMA based BS-selection scheme. In the TDMA based
BS-selection, we also apply the priority BS-selection algorithm
given by Fig.~\ref{fig-priority-selection} when the cardinality $L$
of $\Omega_L$ is specified. Obtaining $\Omega_L$, we further divide
each time frame into $K_{\mathrm{mu}}$ time slots for data
transmissions to $K_{\mathrm{mu}}$ users, respectively. The $n$th
user's time-slot length is set equal to $T\times t_{L,n}$ for
$n=1,2,\ldots,K_{\mathrm{mu}}$, where $t_{L,n}$ is the normalized
time-slot length. Moreover, we still use the probabilistic
transmission strategy across different $\Omega_L$ generated through
Fig.~\ref{fig-priority-selection}, where the probability of using
$\Omega_L$ to transmit data is equal to $\phi_L$. Then, we derive
the TDMA based transmission policies through
solving the following optimization problem $\boldsymbol{A2}$.

\noindent ~$\boldsymbol{A2}$: TDMA based BS-usage minimization
\begin{eqnarray}
&&\hspace{-40pt} \min_{(\boldsymbol{t},\boldsymbol{\phi})}
\left\{\overline{L} \right\}=
\min_{(\boldsymbol{t},\boldsymbol{\phi})}
\left\{\mathbb{E}_{\mathbf{H}}
\left\{\sum_{L=0}^{K_{\mathrm{bs}}}L\phi_L\right\}\right\}\nonumber\\
&&\hspace{-40pt}\mbox{s.t.: } \,1).~
\sum_{L=0}^{K_{\mathrm{bs}}}\phi_L = 1, \quad\quad
\forall\,\mathbf{H},\label{eq-multiple-TDMA-constraint-pt}\\
&&\hspace{-40pt}\quad\quad2).~ \sum_{n=1}^{K_{\mathrm{mu}}}t_{L,n}=
1, \quad\quad
\,\forall\,\mathbf{H},~L=1,2,\ldots,K_{\mathrm{bs}},\label{eq-multiple-TDMA-constraint-td}
\\
&&\hspace{-40pt}\quad\quad3).~\mathbb{E}_{\mathbf{H}}\!\left\{\sum_{L=0}^{K_{\mathrm{bs}}}\phi_Le^{-\theta_n
t_{L,n}R^{(n)}_{\Omega_L}}-e^{-\theta_n\overline{C}_n}\!\!\right\}\!\leq
0, ~\forall\, n, \label{eq-multiple-TDMA-constraint-effcp}
\end{eqnarray}
where $\boldsymbol{\phi}$ and $\boldsymbol{t}$ are functions of
$\mathbf{H}$. In particular, we have
$\boldsymbol{\phi}\triangleq(\phi_0,\phi_1,\phi_2,\ldots,\phi_{K_{\mathrm{mu}}})$,
$\boldsymbol{t}\triangleq(\boldsymbol{t}_1,\boldsymbol{t}_2,\ldots,\boldsymbol{t}_{K_{\mathrm{bs}}})$,
and
$\boldsymbol{t}_L\triangleq(t_{L,1},t_{L,2},\ldots,t_{L,K_{\mathrm{bs}}})$.


\begin{theorem}
Problem $\boldsymbol{A2}$'s optimal solution pair
$(\boldsymbol{t}^*,\boldsymbol{\phi}^*)$, if existing, is determined
by
\begin{eqnarray}
t^*_{L,n} =
\left[\frac{1}{\theta_nR^{(n)}_{\Omega_L}}\log\left(\frac{\lambda_n^*\theta_n
R^{(n)}_{\Omega_L}}{\delta_{\mathbf{H},L}^*}\right)\right]^+,
\label{eq-multiple-TDMA-opt-t}
\end{eqnarray}
for all $L$, $n$, and $\mathbf{H}$, and
\begin{eqnarray}
\phi_L^* \!=\!\left\{\!\!\!
\begin{array}{ll}
1, & \!\!\!\mbox{if} ~ L = \arg\min\limits_{\ell}\left\{\ell
+\sum_{n=1}^{K_{\mathrm{mu}}}\lambda_n^*e^{-\theta_n
t_{L,n}^*R^{(n)}_{\Omega_{\ell}}}\right\}\!;\!\!\!\\
0, & \!\!\!\mbox{otherwise},
\end{array}
\right.\label{eq-multiple-TDMA-opt-phi}
\end{eqnarray}
for all $L$ and $\mathbf{H}$, where $\delta_{\mathbf{H},L}^*$ under
given $\{\lambda_n^*\}_{n=1}^{K_{\mathrm{mu}}}$ is determined by
satisfying $\sum_{n=1}^{K_{\mathrm{mu}}}t_{L,n}^*=1$, and
$\{\lambda_n^*\}_{n=1}^{K_{\mathrm{mu}}}$ needs to be selected such
that the equality of Eq.~(\ref{eq-multiple-TDMA-constraint-effcp})
holds. \label{theorem-TDMA}
\end{theorem}
\begin{proof}
The detailed proof of Theorem~\ref{theorem-TDMA} is omitted due to
lack of space, but is provided on-line in~\cite{Q-Du-online}.
\end{proof}

\section{Simulation Evaluations}
\label{sect-simulations}

\begin{figure*}[t]
\vspace{-6pt}\centerline{\hspace{4pt}\includegraphics[width=2.65in]{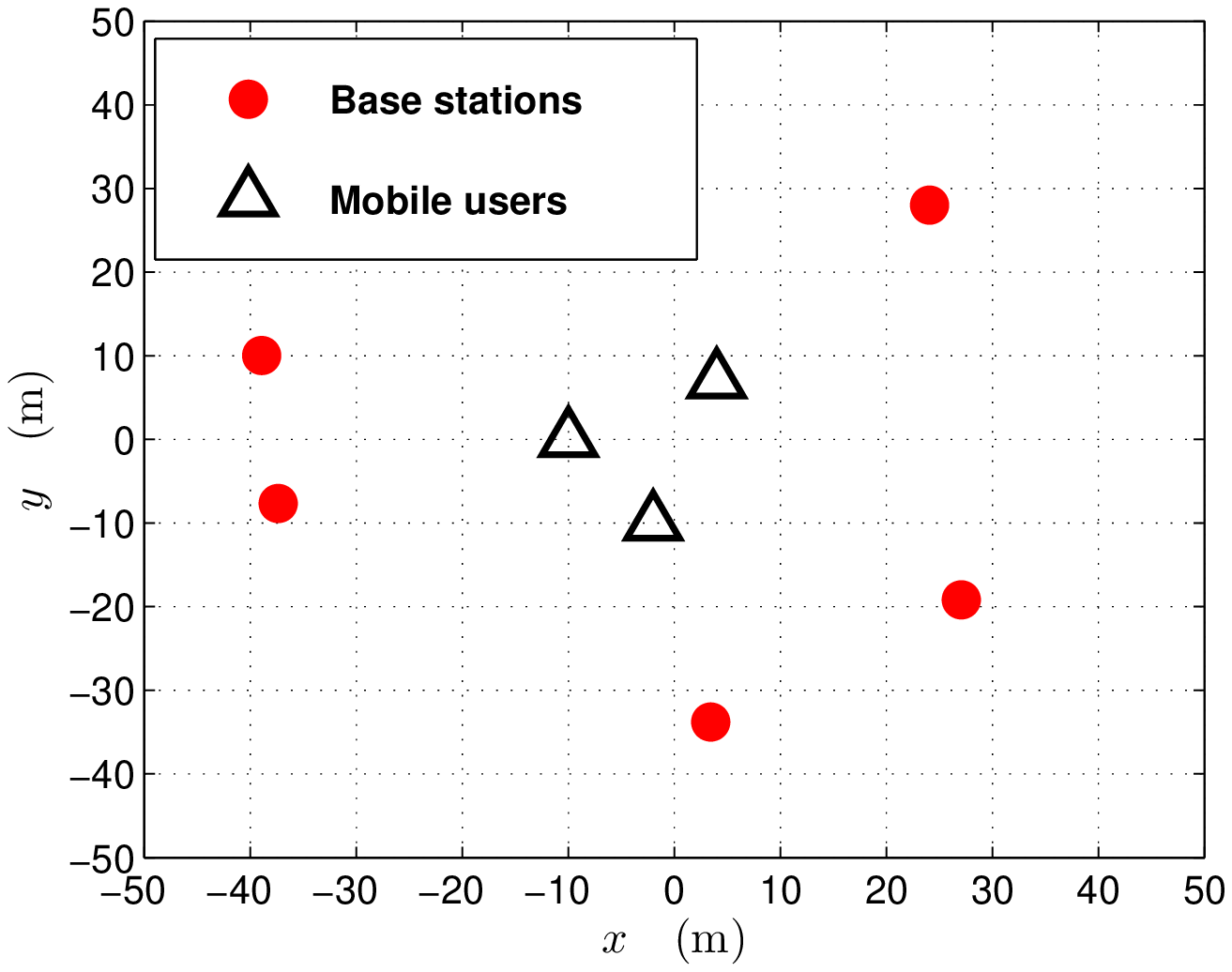}\hspace{-15pt}
\includegraphics[width=2.65in]{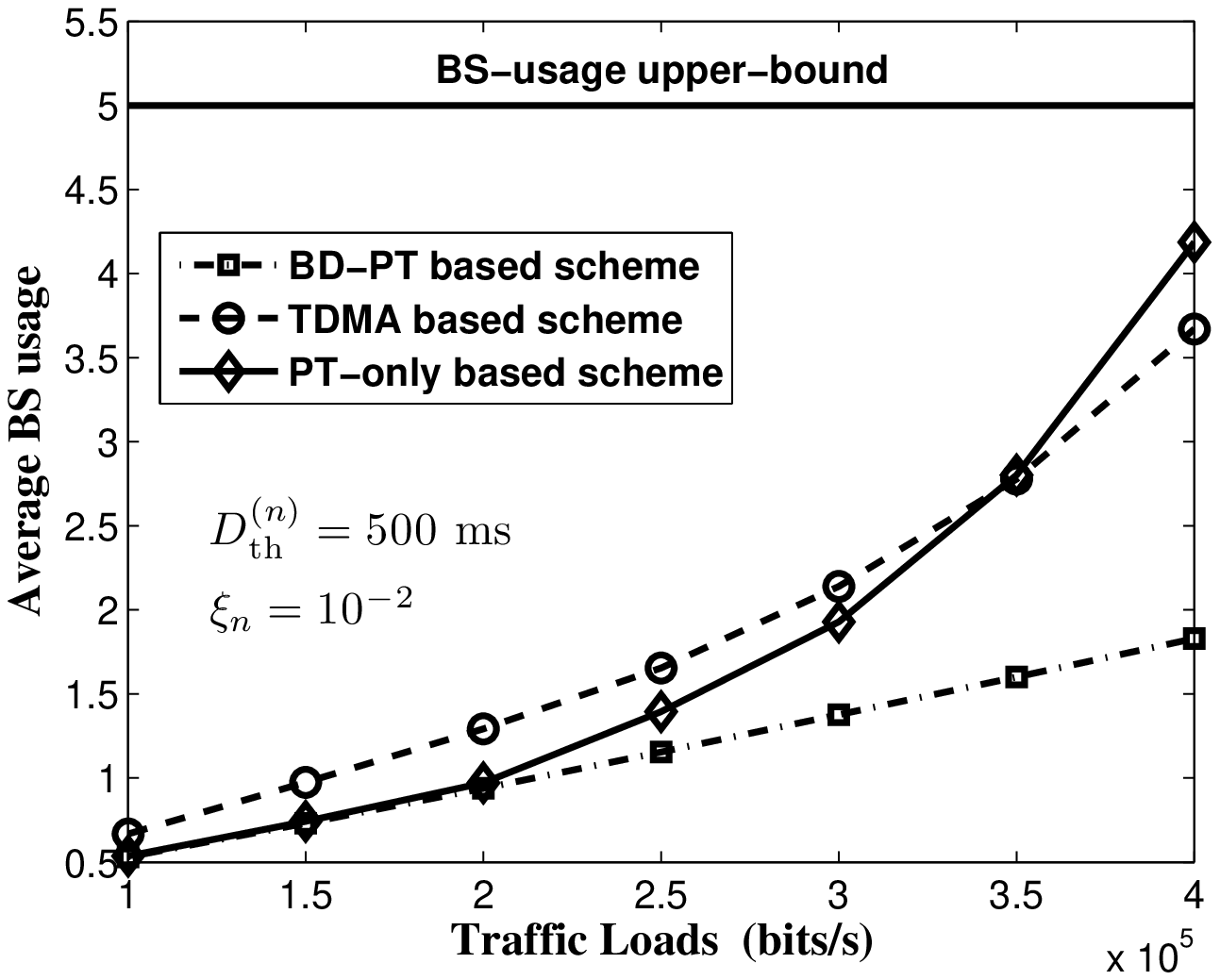}\hspace{-15pt}\includegraphics[width=2.65in]{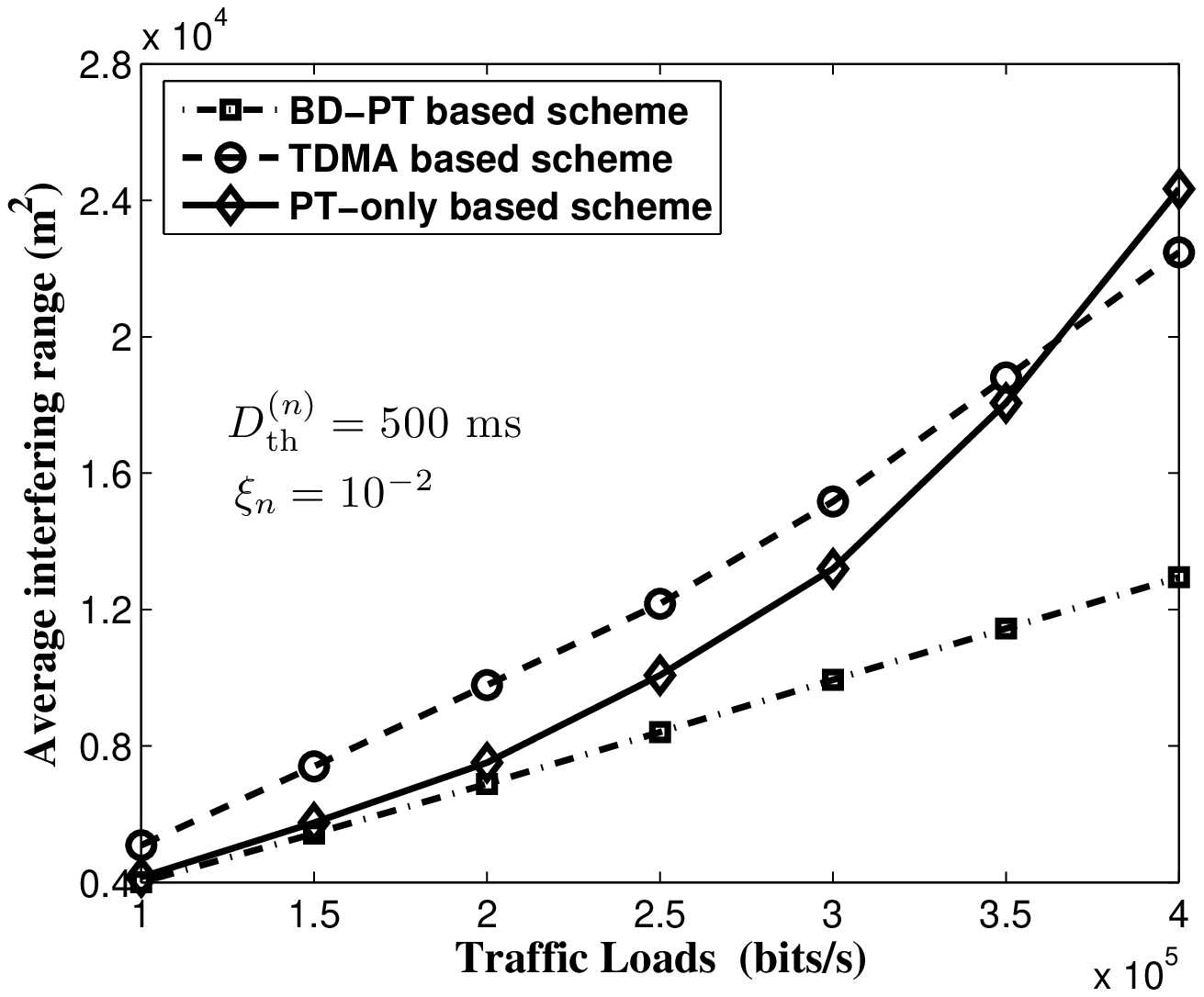}}\vspace{-3pt}
\centerline{~~(a)\hspace{170pt}(b)\hspace{170pt}(c)}\vspace{-7pt}\caption{(a)
The deployment of BS's and the positions of mobile users, where
$K_{\mathrm{mu}}=3$ and $K_{\mathrm{bs}}=5$. (b) Simulation results
of the average BS usage $\overline{L}$ versus traffic load under the
specified delay-QoS requirements, where $\xi_n=10^{-2}$ and
$D_{\mathrm{th}}^{(n)}=500$ ms for all $n$; $M_m=3$; $\kappa=1$. (b)
Simulation results of the average interfering range versus traffic
load under the same system setup as in~(b).}
\vspace{-5pt}\label{fig-group-1}
\end{figure*}
\begin{figure*}[t]
\vspace{-6pt}\centerline{\hspace{4pt}\includegraphics[width=2.65in]{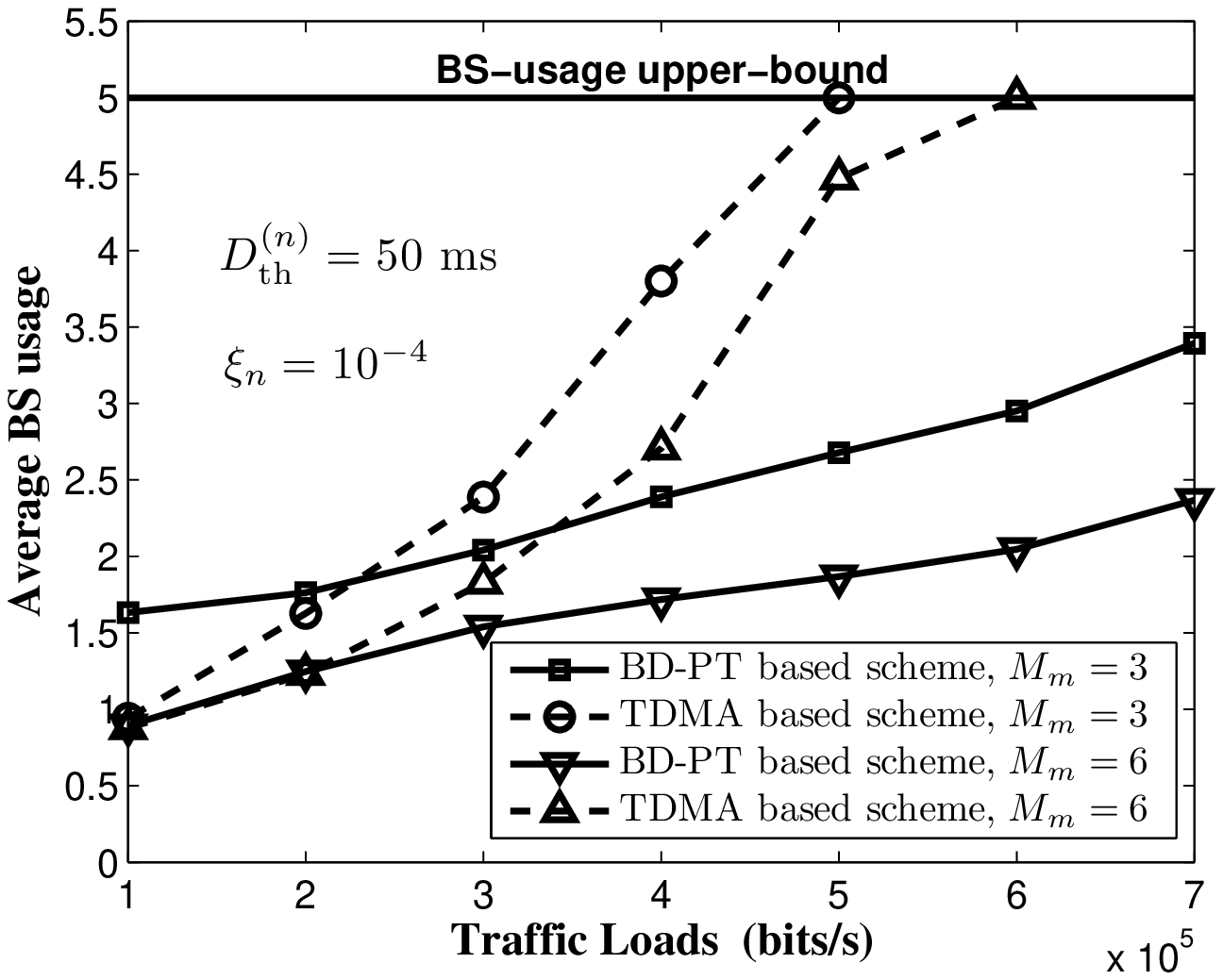}\hspace{-15pt}
\includegraphics[width=2.65in]{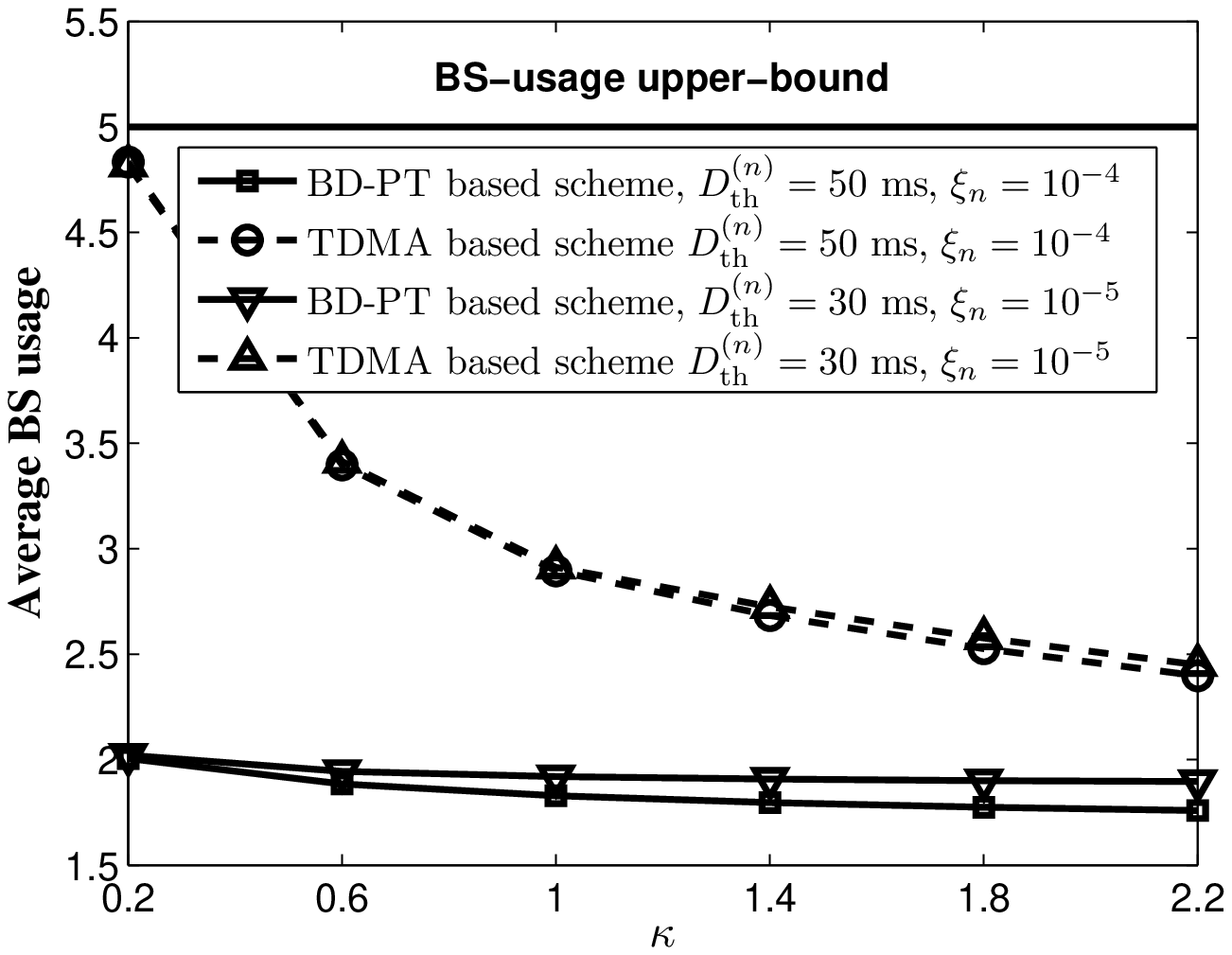}\hspace{-15pt}\includegraphics[width=2.65in]{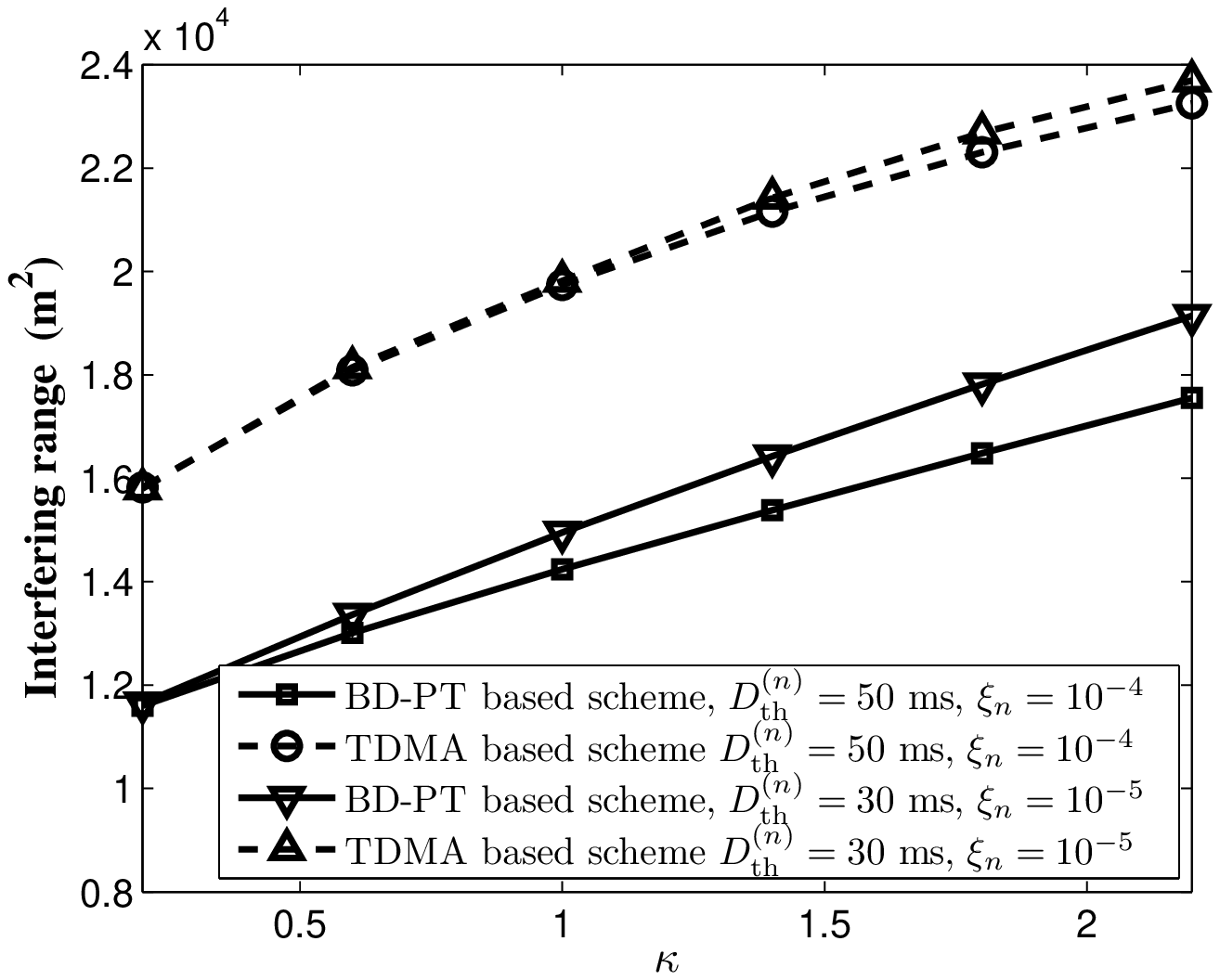}}\vspace{-3pt}
\centerline{~~(a)\hspace{170pt}(b)\hspace{170pt}(c)}\vspace{-7pt}\caption{(a)
The average BS usage $\overline{L}$ versus traffic load under the
specified delay-QoS requirements, where $D_{\mathrm{th}}^{(n)}=50$
ms, $\xi_n=10^{-4}$, and $N_n=2$ for all $n$; $\kappa=1$. (b)
Average BS usage versus $\kappa$, where $M_m=5$. (c) Average
interfering range versus $\kappa$, where the system setup is the
same as in~(b).} \vspace{-8pt}\label{fig-group-2}
\end{figure*}



We use simulations to evaluate the performances of our proposed
QoS-aware BS selection schemes for distributed MIMO links. The BS's
deployment and the mobile users' positions are shown in
Fig.~\ref{fig-group-1}(a), where $K_{\mathrm{bs}}=5$ and
$K_{\mathrm{mu}}=3$. We set $T=10$ ms and $B=10^5$ Hz. We further
assume that all users have the same number of receive antennas, all
distributed BS's have the same number of transmit antennas, and the
incoming traffic loads for all users are equal. Furthermore, we
employ the following average power propagation model. Specifically,
the average received power gain $\overline{h}_{n,m}$ is equal to
$G/d_{n,m}^{\eta}$, where $d_{n,m}$ is the distance between the
$n$th mobile user and the $m$th BS, $G$ is a constant factor, and
$\eta$ is the path loss exponent typically varying from 2 to
6~\cite{T-Rappaport}. Without loss of generality, we let
$\mathcal{P}_{\mathrm{ref}}=1$ and select $G$ such that
$\overline{h}_{n,m}=0$~dB at $d_{n,m}=50$~m. Also, we set
$\sigma_{\mathrm{th}}^2=0$ dB for evaluating of the average
interfering range (see
Section~\ref{sect-sysmodel-design-objective}).

Figures~\ref{fig-group-1}(b) and~\ref{fig-group-1}(c) compare the
average BS usage and interfering range as functions of the incoming
traffic load among our derived QoS-aware BS-selection schemes,
including the joint BD-PT, TDMA, and PT-only based schemes.
Fig.~\ref{fig-group-1}(b) shows that as the traffic load increases,
all scheme's average BS usages become larger to satisfy the more
stringent QoS requirements. However, the TDMA and PT-only based
schemes' BS usages increase much more rapidly than our proposed
BD-PT based scheme. This is because block diagonalization for
multi-user MIMO communications can effectively take advantage of
space multiplexing in removing the cross-interferences among all
mobile users, and thus can achieve high spectral efficiency and
system throughput. We can further observe that when the traffic load
gets lower (larger), the PT-only based scheme needs less (more) BS's
to satisfy the specified QoS requirements, as compared to the TDMA
based scheme. Fig.~\ref{fig-group-1}(c) plots the average
interfering range caused by distributed MIMO transmissions, which
displays the similar results to Fig.~\ref{fig-group-1}(b). This is
expected because the total used power in each fading state linearly
increases with the cardinality $L$ of selected BS-subset, as shown
in Section~\ref{sect-sysmodel-power}.

Figure~\ref{fig-group-2}(a) plots the average BS usage against
traffic load with more stringent QoS constraints than the
constraints used in Fig.~\ref{fig-group-1}(b). Under these more
stringent constraints, the PT-only based scheme cannot support the
specified QoS requirements for the incoming traffics and thus are
not plotted in Fig.~\ref{fig-group-2}(a), which implies that the
PT-only based scheme only works efficiently with loose QoS
constraints.  We can observe from Fig.~\ref{fig-group-2}(a) that the
BD-PT based scheme generally outperforms the TDMA based scheme in
terms of requiring fewer BS's, especially when the traffic load is
high. As shown in Fig.~\ref{fig-group-2}(a), for traffic load higher
than or equal to 500 Kbits/s, the BS usage of the TDMA based scheme
will reach the upper-bound, which is equal to $K_{\mathrm{bs}}$.
This implies that all wireless resources have been used up while the
specified QoS requirements for the incoming traffic still cannot be
satisfied. In contrast, the BD-PT based scheme can clearly support
even higher traffic load. An interesting observation is that the
TDMA based scheme performs slightly better than the BD-PT based
scheme, when the traffic load is low and the number of antennas per
BS is small. This is because the advantage of BD technique can be
effectively used when the spatial-multiplexing degree order is high.
However, clearly the small number of transmit antennas can already
successfully support smaller traffic load through TDMA strategy,
while the BD in this case is not very effective due to the limited
number of transmit antennas, implying insufficient freedom for
spatial multiplexing.

Figures~\ref{fig-group-2}(b) and~\ref{fig-group-2}(c) depict the
average BS usage and interfering range, respectively, versus the
parameter $\kappa$, where $\kappa$ is defined in
Section~\ref{sect-sysmodel-power}, which is the power increasing
rate with the number of BS's selected for distributed MIMO
transmissions. We can see that the average BS usage and the
interfering range of the BD-PT based scheme are much smaller than
those of the TDMA based scheme. As shown in
Figs.~\ref{fig-group-2}(b) and~\ref{fig-group-2}(c), the lower delay
bound and the smaller violation probability threshold, implying more
stringent delay-QoS requirements, cause more BS usage and thus
larger interfering range. This is because in order to satisfy more
stringent QoS requirements, more BS's need to get involved with the
cooperative downlink transmissions to achieve the high system
throughput for all mobile users. This also demonstrates that our
proposed schemes can effectively adjust the transmission strategy to
adapt to the specified QoS requirements. In addition, the average
BS-usage is a decreasing function of $\kappa$ but the interfering
range is an increasing function. This suggests that we can use more
power to tradeoff the lower implementation complexity in distributed
MIMO transmissions.

\section{Conclusions}
\label{sect-conclusion}

We proposed the QoS-aware BS-selection schemes for the distributed
wireless MIMO links, which aim at minimizing the BS usages and
reducing the interfering range, while satisfying diverse statistical
delay-QoS constraints over multiple mobile users. In particular, we
developed the joint block-diagonalization and
probabilistic-transmission based scheme, the TDMA based scheme, and
the pure probabilistic-transmission based scheme, respectively, to
implement efficient BS-selection and the corresponding resource
allocation algorithms for QoS provisioning of mobile users.
Simulation results show that the joint block-diagonalization and
probabilistic-transmission based scheme generally outperforms the
TDMA based and pure probabilistic-transmission based schemes in
terms of requiring less BS's for data transmissions and decreasing
the interfering range caused to the entire wireless networks.
Moreover, the TDMA and probabilistic-transmission based schemes is
efficient when the traffic load is not heavy.

%
%

\end{document}